\documentclass[11pt]{article}
\usepackage[margin=1in]{geometry}
\usepackage{amsfonts,amsmath,amssymb,amsthm}
\usepackage{enumerate}
\usepackage{hyperref}

\usepackage[short]{privbibg}
\usepackage[textsize=tiny]{todonotes}

\newcommand{\PPull}{\mbox{\tt PUSH-PULL}}
\newcommand{\Push}{\mbox{\tt PUSH}}
\newcommand{\Pull}{\mbox{\tt PULL}}

\newtheorem{theorem}{Theorem} 
\newtheorem{corollary}[theorem]{Corollary}
\newtheorem{lemma}[theorem]{Lemma}
\newtheorem{claim}[theorem]{Claim}
\newtheorem{definition}[theorem]{Definition}

\DeclareMathOperator{\polylog}{polylog}
\DeclareMathOperator{\diam}{diam}
\DeclareMathOperator{\vol}{vol}
\DeclareMathOperator{\Exp}{\mathbf{E}}
\DeclareMathOperator{\Var}{\mathbf{Var}}
\DeclareMathOperator{\Cov}{\mathbf{Cov}}

\begin{document}

\title{Tight Bounds for Rumor Spreading with Vertex Expansion}

\author{
George Giakkoupis\\
INRIA, France\\
\href{mailto:george.giakkoupis@inria.fr}{george.giakkoupis@inria.fr}
}

\maketitle

\begin{abstract}
We establish a bound for the classic \PPull{} rumor spreading protocol on arbitrary graphs, in terms of the vertex expansion of the graph.
We show that $O(\log^2(n) /\alpha)$ rounds suffice with high probability to spread a rumor from a single node to all $n$ nodes, in any graph with vertex expansion at least $\alpha$.
This bound matches the known lower bound, and settles the question on the relationship between rumor spreading and vertex expansion asked by Chierichetti, Lattanzi, and Panconesi~\cite{Chierichetti2010soda}.
Further, some of the arguments used in the proof may be of independent interest, as they give new insights, for example, on how to choose a small set of nodes in which to plant the rumor initially, to guarantee fast rumor spreading.
\end{abstract}

\section{Introduction}

We study a classic randomized protocol for information dissemination in networks, known as \emph{(randomized) rumor spreading}.
The protocol proceeds in a sequence of synchronous rounds.\footnote{There are also asynchronous versions of rumor spreading (see, e.g.,~\cite{Boyd2006}), but in this paper we focus on synchronous protocols.}
Initially, in round 0, an arbitrary node learns a piece of information,
the \emph{rumor}.
This rumor is then spread iteratively to other nodes:
In each round, every informed node (i.e., every node that learned the rumor in a previous round) chooses a random neighbor and sends the rumor to that neighbor.
This is the \Push{} version of the protocol.
The \Pull{} version is symmetric:
In each round, every uninformed node contacts a random neighbor, and if this neighbor knows the rumor it sends it to the uninformed node.
Finally, the \PPull{} algorithm is the combination of both:
In each round, every node chooses a random neighbor to send the rumor to, if the node knows the rumor, or to request the rumor from, otherwise.

These protocols were proposed almost thirty years ago, and have been the subject of extensive study, especially in the past decade.
The most studied question concerns the number of rounds that these protocols need to spread a rumor in various network topologies.
It has been shown that $O(\log n)$ rounds suffice with high probability (w.h.p.)  for several families for networks, from basic communication networks, such as complete graphs and hypercubes, to more complex structures, such as preferential attachment graphs modeling social networks (see the Related Work Section).

A main motivation for the study of rumor spreading is its application in algorithms for broadcasting in communication networks~\cite{Demers1987podc,Feige1990,Karp2000focs}.
Rumor spreading provides a scalable alternative to the flooding protocol (where each node sends the information to \emph{all} its neighbors in a round), and a simpler and more robust alternative to deterministic solutions.
The advantages of simplicity (each node makes a simple local decision in each round; no knowledge of the global topology is needed; no state is maintained), scalability (each node initiates just one connection per round), and robustness (the protocol tolerates random node/link failures without the use of error recovery mechanisms) make rumor spreading protocols particularly suited for today's distributed networks of massive-scale.
Such networks, e.g., peer-to-peer, mobile ad-hoc, or sensor networks, are highly dynamic, suffer from frequent link and node failures, or nodes have limited computational, communication, and energy resources.

Another motivation for the study of rumor spreading protocols is that they provide intuition on how information spreads in social networks~\cite{Chierichetti2009icalp}.
More generally, understanding these simple rumor spreading protocols may
lead to a better understanding of more realistic epidemic processes on social and other complex networks.

In this paper, our main focus is the connection between rumor spreading and graph expansion properties of networks.
Many of the topologies for which rumor spreading is known to be fast have high expansion.
Further, empirical studies indicate that social networks have good expansion properties as well~\cite{Eubank04nature,Leskovec08www}.


Several works have studied the relationship of rumor spreading with the \emph{conductance} of the network graph~\cite{Mosk-Aoyama2008,Chierichetti2010soda,Chierichetti2010stoc,Giakkoupis2011stacs}.
The conductance $\phi \in (0,1]$ is a standard expansion measure defined roughly as the minimum ratio of the edges leaving a set of nodes over the total number of edges incident to these nodes~(see Section~\ref{sec:preliminaries}).
The main result of the above works is an upper bound of $O(\log(n)/\phi)$ rounds for \PPull{} to inform all nodes w.h.p., for any $n$-node graph with conductance at least $\phi$.
This bound is tight, as there are graphs with diameter $\Omega(\log(n)/\phi)$~\cite{Chierichetti2010stoc}.

The above result has been used in the design of a recent breed of information dissemination protocols~\cite{Censor-Hillel2011soda,Censor-Hillel2012stoc,Haeupler2013soda}.
These protocols rely on the fact that \PPull{} spreads information fast in subgraphs of high conductance, and they combine \PPull{} with more sophisticated rules on how each node chooses the neighbor to contact in each round.
These new protocols achieve fast information spreading in a broader class of networks, and some achieve for all graphs time bounds that are close (within poly-logarithmic factors) to the network diameter, which is the natural lower bound for information dissemination in networks.

More recently, another standard measure of expansion, vertex expansion, has been studied in connection with rumor spreading.
The vertex expansion $\alpha \in (0,1]$ of a graph is, roughly, the minimum ratio of the neighbors that a set of nodes has (that are not in the set) over the size of the set.
In general, vertex expansion is incomparable to conductance, as there are graphs with high vertex expansion but low conductance,
and vice versa.
(For an account of the differences between the two measures see~\cite{GiakkoupisS2012soda}.)
The question of whether high vertex expansion implies fast rumor spreading (similarly to high conductance) was highlighted as an interesting open problem in~\cite{Chierichetti2010soda}.
This problem was studied in~\cite{Sauerwald2011soda,GiakkoupisS2012soda}, and their main result was an upper bound for \PPull{} of $O(\log^{2.5}(n)/\alpha)$ rounds w.h.p., for any graph with vertex expansion at least $\alpha$.
The precise bound is $O(\log n \cdot \log \Delta \cdot \sqrt{ \log (2 \Delta/ \delta)}/\alpha)$, where $\Delta$ and $\delta$ are the maximum and minimum node degrees, respectively.
Further, a lower bound of $\Omega(\log n \cdot \log(\Delta)/\alpha)$ was shown,  assuming $\Delta/\alpha\leq n^{1-\epsilon}$ for some constant $\epsilon > 0$.

\paragraph{Our Contribution.}

Our main result is the following upper bound for \PPull{} in terms of vertex expansion that matches the known lower bound.
\begin{theorem}
    \label{thm:ub-pushpull}
    Let $G=(V,E)$ be a graph with $|V| = n$, maximum degree at most $\Delta$, and vertex expansion at least $\alpha$.
    For any such graph $G$ and constant $\beta>0$, with probability $1-O(n^{-\beta})$ \PPull{} informs all nodes of $G$ in $O\big(\log n \cdot \log (\Delta)/\alpha\big)$ rounds.
\end{theorem}

This result, together with the $O(\log(n)/\phi)$ bound with conductance, resolve completely the natural question asked by Chierichetti, Lattanzi, and Panconesi~\cite{Chierichetti2009icalp,Chierichetti2010soda}, on the relationship between rumor spreading and the two most standard measures of graph expansion.

Our proof of Theorem~\ref{thm:ub-pushpull} can be summarized as follows.
Let $S$ be the set of informed nodes after a given round, and $\partial S$  be the boundary of $S$, i.e., the set of nodes from $V-S$ that have some neighbor in $S$.
The proof defines a new simple measure of the expansion of $S$, called boundary expansion.
If $S$ has low boundary expansion, then we prove that a constant fraction of the boundary $\partial S$ gets informed in an expected number of $O(\log\Delta)$ rounds; this is the core argument of the proof.
If, instead, $S$ has high boundary expansion, then we have that an expected number of $\Omega(|\partial S|)$ nodes from $V-(S\cup\partial S)$ are added to the boundary in a single round.
It follows that a simple potential function $\Psi(S)$ that counts 1 for each informed node and 1/2 for each node in the boundary, increases ``on average" per round by at least $\Omega(|\partial S|/\log\Delta) = \Omega(\alpha\Psi(S)/\log\Delta)$; this is the right increase rate we need to show the $O\big(\log n \cdot \log (\Delta)/\alpha\big)$ time bound.

The above proof borrows some ideas from the analysis in~\cite{GiakkoupisS2012soda}. In particular, both proofs study the growth of essentially the same potential function~$\Psi$.
However, the arguments they use are different.
Moreover, our proof provides a clear intuition for the result not conveyed by earlier proofs.

Some of the arguments in our proof may be of independent interest as tools for the analysis of rumor spreading.
We demonstrated this by reusing those arguments to show the following smaller results.
\begin{enumerate} 
\item \label{res:dominatingset}
Our first result serves as a ``warm-up" for the main proof, as it uses a simpler version of the core argument of our analysis.
We show that if a rumor is initially known by a subset of nodes that is a dominating
set,
then $O(\log n)$ rounds of \PPull{} suffice w.h.p.\ to inform the other nodes.
(In fact we can use just \Pull{} instead of \PPull.)
This result is somewhat relevant to problems in viral marketing~\cite{Domingos2001kdd,Kempe2003kdd}:
It says that if we want to plant a rumor, or ad, in a (small) initial set of nodes in an arbitrary network, so that the remaining nodes get informed quickly by rumor spreading, then it suffices that the set we choose be a dominating set.\footnote{Note that the analysis of the more sophisticated information dissemination algorithm proposed in~\cite{Censor-Hillel2012stoc,Haeupler2013soda}, yields time bounds of $O(\log^3 n)$ and $O(\log^2 n)$, respectively, for this problem.}

\item \label{res:diametertwo}
Our next result uses (in an interesting way) the main argument from the proof of Result~1 above, to show that \PPull{} spreads a rumor in $O(\log n)$ rounds w.h.p.\ in any graph of diameter (at most) 2.
This result can be viewed as an extension to the classic result that rumor spreading takes $O(\log n)$ rounds in graphs of diameter~1, i.e., in complete graphs.
Note that unlike complete graphs, some graphs of diameter 2 have bad expansion.
Also the result does not hold for graphs of diameter~3.


\item \label{res:regular}
For regular graphs, a variant of our analysis yields an upper bound in terms of a new natural expansion measure, and this bound is strictly stronger than the $O(\log(n)/\phi)$ bound with conductance.
The new measure, denoted $\xi$, is defined as the minimum over all sets $S$ with $|S|\leq n/2$, of the product $\alpha(S)\cdot\phi(\partial S)$ of the vertex expansion of $S$ and the conductance of its boundary $\partial S$.
The bound we show is $O(\log(n)/\xi)$ w.h.p.\ for any regular graph.

\end{enumerate}
The proofs of Results~\ref{res:diametertwo} and~\ref{res:regular} can be found in the Appendix, in Sections~\ref{sec:dimeter-two} and~\ref{sec:regular}, respectively.

\paragraph{Related Work.}

The first works on rumor spreading provided a precise analysis of \Push{} on complete graphs~\cite{Frieze1985, Pittel1987}.
Time bounds of $O(\log n)$ rounds were later proved for hypercubes and random graphs~\cite{Feige1990}.
Other symmetric graphs similar to the hypercube in which rumor spreading takes $O(\log n + \diam)$ rounds were studied in~\cite{Elsasset2007stacs}.
A refined analysis for random graphs proving essentially the same time bound as for complete graphs was provided in~\cite{Fountoulakis2010infocom}, and extended to random regular graphs in~\cite{Fountoulakis2010random}.
The authors of~\cite{Chierichetti2009icalp} studied rumor spreading on preferential attachment graphs, which are used as models for social networks, and showed that \Push{} or \Pull{} need polynomially many rounds, whereas \PPull{} needs only $O(\log^2 n)$ rounds.
The last bound was subsequently improved to $\Theta(\log n)$~\cite{Doerr2011stoc}.
Another class of graphs used to model social networks was considered in~\cite{Fountoulakis2012soda}, and it was shown that \PPull{} needs just $\Theta(\log\log n)$ rounds to inform all but an $\epsilon$-fraction of nodes.

For general graphs, it was shown in~\cite{Mosk-Aoyama2008} a bound of $O(\log(n)/\Phi)$ for a version of \PPull{} with non-uniform probabilities for neighbor selection, where $\Phi$ is the conductance of the matrix of selection probabilities (which is different than the conductance $\phi$ of the graph).
This result does not extend to standard \PPull{}.
A comparable bound in terms of the mixing time of an appropriate random walk was shown in~\cite{Boyd2006}.
For \PPull{}, a polynomial bound in $\log(n)/\phi$ was shown in~\cite{Chierichetti2009icalp} via a connection to a spectral sparsification process.
An improved, almost tight bound was shown in~\cite{Chierichetti2010stoc}, and the tight $O(\log(n)/\phi)$ bound was shown in~\cite{Giakkoupis2011stacs}.
For \Push{} or \Pull{} this bound hold for regular graphs, but not for general graphs.


The bound for \PPull{} with conductance has been used in subsequent works~\cite{Censor-Hillel2010podc,Censor-Hillel2011soda,Censor-Hillel2012stoc,Haeupler2013soda}, mainly to argue that rumors spread fast in subgraphs of high conductance.
A refinement of conductance, called weak conductance, which is greater or equal to $\phi$ was introduced in~\cite{Censor-Hillel2010podc} and related to the time for \PPull{} to inform a certain fraction of nodes.
A gossip protocol for the problem in which every node has a rumor initially, and must receive the rumors of all other nodes was proposed in~\cite{Censor-Hillel2011soda}.
The protocol alternates rounds of \PPull{} with rounds of deterministic communication, and guarantees fast information spreading in all graphs with high weak conductance.
In~\cite{Censor-Hillel2012stoc,Haeupler2013soda} protocols for the same gossip problem were proposed that need only $O(\diam\cdot\polylog(n))$ or $O(\diam + \polylog(n))$ rounds.
One of these protocols is even deterministic~\cite{Haeupler2013soda}.
Hopefully, the results presented in the current paper will also help in the design of new protocol for information dissemination.

\section{Notation}
\label{sec:preliminaries}

Throughout the paper we assume that graph $G = (V,E)$ is connected, and $n=|V|$.
For a node $v\in V$, we denote by $N(v)$ the set of $v$'s neighbors in $G$, and $\deg(v) = |N(u)|$ is the degree of $v$.
By $\Delta$ and $\delta$ we denote the maximum and minimum node degrees of $G$.
For a set of nodes $S\subseteq V$, we denote by $\partial S$ the \emph{boundary} of $S$, that is, the set of neighbors of $S$ that are not in $S$; formally, $\partial S = \{ v \in V - S \colon N(v)\cap S\neq \emptyset\}$.
We will write $S^+$ to denote the set $S\cup \partial S$.
Set $S$ is a \emph{dominating set} of $G$ iff each node $v\in V$ either belongs to $S$ or has a neighbor from $S$, i.e., $S^{+} = V$.
The vertex expansion of a non-empty set $S\subseteq V$ is
$
    \alpha(S) = {|\partial S|}/{|S|},
$
and the vertex expansion of $G$ is
\[
    \alpha(G) = \min_{S\subset V,\ 0<|S|\leq n/2} \alpha(S).
\]
The volume of $S$ is $\vol(S) = \sum_{v\in S}\deg(v)$.
By $E(S,V-S)$ we denote the set of edges with one endpoint from $S$ and the other from $V-S$.
The conductance of a non-empty set $S\subseteq V$ is
$
    \phi(S) = {|E(S,V-S)|}/{\vol(S)},
$
and the conductance of $G$ is
\[
    \phi(G) = \min_{S\subset V,\ 0<\vol(S)\leq \vol(V)/2} \phi(S).
\]
We have $0< \phi(G),\alpha(G) \leq 1$, and $(\delta/\Delta)\cdot \phi(G) \leq \alpha(G)\leq \Delta\cdot\phi(G)$~\cite{GiakkoupisS2012soda}.

\section{Warm-up: Rumors Spread Fast from a Dominating Set}

We show the following result in this section.
\begin{theorem}
    \label{thm:dominating-set}
    Let $S\subseteq V$ be a dominating set of $G=(V,E)$, i.e., $S^+:= S\cup\partial S = V$.
    Suppose that all nodes $u\in S$ know of a rumor initially.
    Then \Pull{} informs all remaining nodes  in $O(\log n)$ rounds, with probability $1-O(n^{-\beta})$ for any constant $\beta>0$.
\end{theorem}

We start with an overview of the proof.
From a standard lemma on the symmetry between \Push{} and \Pull{} (Lemma~\ref{lem:symmetry}), it follows that to bound the time for \Pull{} to spread a rumor from $S$ to a given node $u\notin S$, it suffices to bound instead the time for \Push{} to spread a rumor from $u$ to \emph{some} node in $S$.
We  bound the latter time as follows (Lemma~\ref{lem:push-to-S}).
Let $I_t^u$ be the set of informed nodes after $t$ rounds of \Push{}, when the rumor starts from $u$.
We consider the earliest round $\tau$ for which the expected growth of $I_t^u$ in the next round $\tau + 1$ is smaller than by a constant factor $\epsilon>0$, i.e., $\Exp[|I_{\tau+1}^u-I_\tau^u| \mid I_\tau^u] < \epsilon\cdot |I_\tau^u|$.
We argue that $\tau = O(\log n)$ w.h.p.\ (Claim~\ref{clm:tau}), which is intuitively clear, as up to round $\tau$ the number of informed nodes increased by a constant factor in expectation per round.
Next we bound the harmonic mean of the degrees of nodes in $I_\tau^u$, precisely, we show that $\sum_{v\in I_t^u} \deg(v)^{-1} = \Omega(1)$.
If $I_\tau^u\cap S =\emptyset$ then each node from $I_\tau^u$ has some neighbor in $S$, and from the above bound on the degrees it follows that the probability  in one round that some nodes from $I_\tau^u$ pushes the rumor to a neighbor in $S$ is $\Omega(1)$.
Thus, within $O(\log n)$ rounds after round $\tau$ the rumor has reached some node in $S$ w.h.p.

Next we give the detailed proof.
\begin{lemma}
    \label{lem:symmetry}
    Let $T_{push}(V_1,V_2)$, for $V_1,V_2\subseteq V$, be the number of rounds for \Push{} until a rumor that is initially known to all nodes $u\in V_1$ (and only them) spreads to at least one node $v\in V_2$; let also $T_{pull}(V_1,V_2)$ be defined similarly.
    Then, for any  $V_1,V_2\subseteq V$, the random variables $T_{push}(V_1,V_2)$ and $T_{pull}(V_2,V_1)$ have the same distribution.
\end{lemma}

The proof of Lemma~\ref{lem:symmetry} is essentially the same as that of~\cite[Lemma~3]{Chierichetti2010stoc}, and is therefore omitted.

Next we state our main lemma.
\begin{lemma}
    \label{lem:push-to-S}
    Let $S$ be a dominating set of $G=(V,E)$, and let $u\in V-S$.
    Using \Push, a rumor originated at $u$ spreads to at least one node $v\in S$ in $O(\log n)$ rounds, with probability $1-O(n^{-\beta})$ for any constant $\beta > 0$.
\end{lemma}

The proof of Lemma~\ref{lem:push-to-S} is given in Section~\ref{sec:push-to-S}.

From the two lemmata above, Theorem~\ref{thm:dominating-set} follows easily:
From Lemma~\ref{lem:push-to-S} we obtain that a rumor originated at a given $u\in V-S$ spreads to at least one $v\in S$ after $O(\log n)$ rounds of \Push, with probability $1-O(n^{-\beta-1})$.
Lemma~\ref{lem:symmetry} then implies that a rumor known to all $v\in S$ reaches $u$ after $O(\log n)$ rounds of \Pull, with the same probability, $1-O(n^{-\beta-1})$.
Applying now the union bound over all $u\in V-S$ yields the claim.

\subsection{Proof of Lemma~\ref{lem:push-to-S}}

\label{sec:push-to-S}

Let $I_t^u$ be the set of informed nodes after $t$ rounds of \Push{}, when the rumor starts from node $u\in V-S$.
Let $\tau$ be the earliest round such that the expected increase of $I_t^u$ in the next round is smaller than $\epsilon\cdot|I_t^u|$, i.e.,
\[
    \tau
    =
    \min\{t \colon \Exp[|I_{t+1}^u| \mid I_t^u] < (1+\epsilon)\cdot |I_t^u|\},
\]
for some positive constant $\epsilon <  1$.
\begin{claim}
    \label{clm:tau}
    With probability $1-n^{-\beta}$ we have $\tau = O(\log n)$.
\end{claim}

The proof of Claim~\ref{clm:tau} is by standard probabilistic arguments, and can be found in the Appendix.

The next claim bounds the harmonic mean of the degrees of nodes $v\in I_\tau^u$.
\begin{claim}
    \label{clm:degrees}
    If $\Exp[|I_{t+1}^u| \mid I_t^u] < (1+\epsilon)\cdot |I_t^u|$ then
    $\sum_{v\in I_t^u} \deg(v)^{-1} \geq 1-\sqrt\epsilon$.\footnote{An equivalent statement for $\sum_{v\in I_t^u} \deg(v)^{-1} \geq 1-\sqrt\epsilon$ is that the harmonic mean of $\deg(v)$, over all $v\in I_t^u$, is at most $|I_t^u|/(1-\sqrt\epsilon)$.}
\end{claim}
\begin{proof}
The proof is by contradiction.
Fix the set $I_t^u$, and let $k=|I_t^u|$.
For $i=1,\ldots, k$, let $u_i$ denote the $i$-th node from $I_t^u$, and let $d_i = \deg(u_i)$.
We will assume that $\sum_{i=1}^k d_i^{-1} < 1-\sqrt\epsilon$ and prove that  $\Exp[|I_{t+1}^u|] \geq (1+\epsilon)\cdot k$.

We count the number of uninformed node that in round $t+1$ receive exactly one copy of the rumor.
This is clearly a lower bound on the number of nodes that get informed in round $t+1$.
Let $d'_i = |N(u_i)\cap I_t^u|$ be the number of neighbors that $u_i$ has in $I_t^u$.
The probability that in round $t+1$ node $u_i$ pushes the rumor to some uninformed node is then $1-d'_i/d_i$.
And if this happens, the probability that the recipient node  does not receive the rumor from any other node in the same round is at least
\[
    \prod_{j\neq i}{(1 - 1/d_j)}
    \geq
    1 - \sum_{j\neq i}(1/d_j)
    \geq
    1 - (1-\sqrt\epsilon)
    =
    \sqrt\epsilon.
\]
Thus the probability that $u_i$ sends the rumor to an uninformed node that does not receive another copy of the rumor is at least $(1-d'_i/d_i)\cdot \sqrt\epsilon$.
Hence, the expected number of uninformed nodes that receive exactly one rumor copy in round $t+1$ is at least $\sum_{i=1}^k\big((1-d'_i/d_i)\cdot \sqrt\epsilon\big)$.
And since
\begin{align*}
    \sum_{i=1}^k(1-d'_i/d_i)
    &=
    k - \sum_{i=1}^k(d'_i/d_i)
    \\&
    \geq
    k - \sum_{i=1}^k (k/d_i)
    \geq
    k - k(1-\sqrt\epsilon)
    =
    k\sqrt\epsilon,
\end{align*}
we obtain a lower bound of $k\sqrt\epsilon\cdot \sqrt\epsilon = k \epsilon$ on the expected number of uninformed nodes that receive exactly one copy of the rumor in round $t+1$.
Hence, the same lower bound holds for the total number on nodes informed in the round, and we conclude that $\Exp[|I_{t+1}^u|] \geq (1+\epsilon)\cdot k$.
\end{proof}

Using Claim~\ref{clm:degrees} it is easy to show an $O(\log n)$ bound w.h.p.\ on the number of additional rounds after round $\tau$, until the rumor spreads to at least one node from $S$.
Suppose that $I_\tau^u = U$ for some set $U\subseteq V-S$.
(If $U\nsubseteq V-S$ then some node from $S$ is already informed.)
Each node $v\in U$ has at least one neighbor in the dominating set $S$, and thus the probability that none of these nodes pushes the rumor to a neighbor in $S$ in a given round $t > \tau$, is upper bounded by
\[
    \prod_{v\in U}\left(1 - \deg(v)^{-1}\right)
    \leq
    e^{-\sum_{v\in U}\deg(v)^{-1}}
    \leq
    e^{-1+\sqrt\epsilon},
\]
by Claim~\ref{clm:degrees}.
This probability bound holds for each round $t>\tau$ independently of the outcome of previous rounds.
It follows that $\ell := (\beta\cdot \ln n)/(1-\sqrt\epsilon) = O(\log n)$ additional rounds after round $\tau$ suffice to spread the rumor to a node in $S$ with probability $1 - (e^{-1+\sqrt\epsilon})^\ell = 1 - n^{-\beta}$.
Combining this with Claim~\ref{clm:tau}, which bounds $\tau$ by $O(\log n)$  with probability $1 - n^{-\beta}$, and applying the union bound gives that the rumor spreads from $u$ to at least one $v\in S$ in $O(\log n)$ rounds with probability $1 - 2n^{-\beta}$.
This completes the proof of Lemma~\ref{lem:push-to-S}.

Another application of the above argument is described in Section~\ref{sec:dimeter-two} of the Appendix.

\section{Bound with Vertex Expansion}
\label{sec:ub-vexpansion}

In this section we prove our main result, Theorem~\ref{thm:ub-pushpull}.

We start with an overview of the proof.
Let $I_t$ denote the set of informed nodes after the first $t$ rounds.
We study the growth of the quantity
$\Psi_t := |I_t| + |\partial I_t|/2 = (|I_t| + |I_t^+|)/2$.
You can think of $\Psi_t$ as a potential function:
each informed node has as a potential of 1, each uninformed node with an informed neighbor has potential 1/2, and the remaining uninformed nodes have potential zero; $\Psi_t$ is then the total potential after round $t$.
We have $1 < \Psi_t \leq n$.
To prove the theorem, we show (as outlined below) that the expected number of rounds needed to double $\Psi_t$  is bounded by $O(\log(\Delta)/\alpha)$, as long as $|I_t| \leq n/2$.
It follows that $O(\log n \cdot \log(\Delta)/\alpha)$ rounds suffice w.h.p.\ to inform $n/2  + 1$ nodes, and by a symmetry argument,  $O(\log n \cdot \log(\Delta)/\alpha)$ additional rounds suffice to inform all remaining nodes w.h.p.

For $\Psi_t$ to double in  $O(\log(\Delta)/\alpha)$ rounds, it suffices that it increases by $\Omega(|\partial I_t|/\log\Delta)$ ``on average" per round, as $|\partial I_t|\geq \alpha\cdot |I_t|$ and thus $|\partial I_t| = \Omega(\alpha\Phi_t)$.
Such an increase can be achieved either by informing $\Omega(|\partial I_t|/\log\Delta)$ nodes from $\partial I_t$, or by informing fewer nodes which however have a total number of $\Omega(|\partial I_t|/\log\Delta)$ neighbors in $\partial(I_t^+)$.
Along this intuition, we distinguish the following two cases, in terms of a simple expansion measure we define for $I_t$, called boundary expansion (Definition~\ref{def:bexpansion}).

The first case is when the boundary expansion of $I_t$ is low (upper-bounded by a constant $\epsilon_h < 1$).
This is the more challenging case, and is the core of our analysis.
Our main lemma in this case is Lemma~\ref{lem:growthIt}, which establishes that a constant fraction of the boundary $\partial I_t$ gets informed in an expected number of $O(\log\Delta)$ rounds.
The proof builds upon and extends the ideas used in the proof of Theorem~\ref{thm:dominating-set}.
We note that in the setting of Theorem~\ref{thm:dominating-set}, the boundary expansion of the set $S$ of informed nodes is zero.

The second case is when the boundary expansion of $I_t$ is high (lower-bounded by a constant $\epsilon_h > 0$).
Then from our definition of boundary expansion it follows that the expected number of nodes from $\partial(I_t^+)$ that have an informed neighbor after the next round is $\Omega(|\partial I_t|)$, i.e., $\Exp[|I_{t+1}^+ - I_{t}^+| \mid I_t] = \Omega(|\partial I_t|)$.
Our main lemma in this case is Lemma~\ref{lem:growthItplus}, which turns the above lower bound on the expected per round growth of $I_{t}^+$ into an upper bound on the expected number of rounds until $I_{t}^+$ grows by some quantity $b$, which depends on the degrees of nodes in $\partial I_t$.

Finally we bound the expected time needed to double $\Psi_t$, in Claim~\ref{clm:psi-double}, by combine the results of the two cases above and using an inductive argument.

The rest of this section is structured as follows.
We define the measure of boundary expansion in Section~\ref{sec:bexpansion}.
In Section~\ref{sec:It} we prove Lemma~\ref{lem:growthIt}, which lower-bounds the growth of $I_t$ when boundary expansion is low.
In Section~\ref{sec:It} we prove Lemma~\ref{lem:growthItplus}, which lower-bounds the growth of $I_t^+$ when boundary expansion is high.
And in Section~\ref{sec:proof-main} we put the pieces together to prove Theorem~\ref{thm:ub-pushpull}.

\subsection{Boundary Expansion}
\label{sec:bexpansion}

\begin{definition}
\label{def:bexpansion}
Let $S\subset V$ be a non-empty set of nodes.
Let $U$ be a random subset of $\partial S$ such that each node $u\in \partial S$ belongs to $U$ with probability  $1/\deg(u)$ independently of the other nodes.
The \emph{boundary expansion $h(S)$ of $S$} is the ratio of the expected number of nodes $v\in \partial(S^+)$ that have some neighbor in $U$, over the size of $\partial S$, i.e.,
\[
    h(S) = \Exp\left[|\{v\in \partial(S^+) \colon N(v)\cap U\neq \emptyset \}|\right]/|\partial S|.
\]
\end{definition}

It follows that
\begin{equation}
    \label{eq:hofS}
    h(S) =
    \frac1{|\partial S|}
    \sum_{v\in \partial(S^+)}
    \bigg(1-\prod_{u\in N(v)\cap\partial S} \left(1-\deg(u)^{-1}\right)\bigg).
\end{equation}
We have $0\leq h(S) < 1$.
The lower bound of 0 is matched iff $S^+ = V$; and
the upper bound holds because the expected number of nodes from $\partial(S^+)$ that have some neighbor in $U$ is upper-bounded by the expected number of the edges between $U$ and $\partial(S^+)$, which is
\[
    \sum_{u\in\partial S} \frac{|N(u)-S^+|}{\deg(u)}
    \leq
    \sum_{u\in\partial S} \frac{\deg(u)-1}{\deg(u)}
    <
    |\partial S|,
\]
and thus $h(S) < 1$.
We note that it is possible to have $h(S)\leq \epsilon_h < 1$, for some constant $\epsilon_h$, even if each node $u\in \partial S$ has $\deg(u)-1$ neighbors in $\partial(S^+)$, if the nodes $u\in \partial S$ have sufficiently many \emph{common} neighbors in $\partial(S^+)$.

We observe that if $I_t=S$, then $h(S)\cdot |\partial S|$ is a lower bound on the expected number of new nodes that have an informed neighbor after round $t+1$, i.e.,
\begin{equation}
    \label{eq:pull-bexpansion}
    \Exp[|I_{t+1}^+ - I_{t}^+| \mid I_t = S]
    \geq
    h(S)\cdot |\partial S|.
\end{equation}
This follows because each node $u\in \partial I_t$ pulls the rumor from $I_t$ in round $t+1$ with probability at least $1/\deg(u)$.

In Section~\ref{sec:Itplus} we will need the following refined definition, which describes the boundary expansion of $S$ contributed by a given subset $T$ of $\partial S$.
\begin{definition}
\label{def:bexpansion2}
Let $S\subset V$ and $T\subseteq \partial S$.
Let $U_T$ be a random subset of $T$ such that each node $u\in T$ belongs to $U_T$ with probability  $1/\deg(u)$ independently of the other nodes.
The \emph{boundary expansion of $S$ due to $T$} is
\[
    h_T(S) = \Exp\left[|\{v\in \partial(S^+) \colon N(v)\cap U_T\neq \emptyset \}|\right]/|\partial S|.
\]
\end{definition}

For $T=\partial S$, the above definition is identical to Definition~\ref{def:bexpansion}, i.e., $h_{\partial S}(S) = h(S)$.

\subsection{The Case of Low Boundary Expansion: Analysis of the Growth of
\texorpdfstring{$\boldsymbol{I_t}$}{It}}
\label{sec:It}

In this section we prove the following result, which is the core lemma of our analysis.
\begin{lemma}
    \label{lem:growthIt}
    Suppose that $I_t = S$ for some set $S \subset V$ with boundary expansion
    $h(S)\leq\epsilon_h$, where $0\leq \epsilon_h < 1$ is an arbitrary constant.
    There is a constant $\epsilon = \epsilon(\epsilon_h) > 0$ such that the expected number of rounds until $\epsilon\cdot |\partial S|$ nodes from $\partial S$ get informed is $O(\log\Delta)$.
\end{lemma}

We start with an overview of the proof.
Similarly to the proof of Theorem~\ref{thm:dominating-set}, to bound the time needed for a given node $u\in\partial S$ to get informed, we bound instead the time needed for a rumor originated at $u$ to spread to some node from $S$ (Lemma~\ref{lem:ppsymmetry}).
Establishing this bound, however, is more difficult in the current setting than in the setting of Theorem~\ref{thm:dominating-set}.
Recall that in the proof of Theorem~\ref{thm:dominating-set}, to bound the time until a rumor originated at $u\in\partial S$ reaches $S$, we first bound the time until the set $I_t^u$ of informed nodes stops doubling (Claim~\ref{clm:tau}), and then bound the harmonic mean of the degrees of nodes in $I_t^u$ at that time (Claim~\ref{clm:degrees});
the bound on the degrees implies that if $I_t^u\cap S=\emptyset$, then with large probability some node from $I_t^u$ will send the rumor to a neighbor in $S$.
This last statement depends critically on the assumption that $S$ is a dominating set, and thus every node from $I_t^u$ has a neighbor in $S$.
This in not true, however, in the current setting, hence the above degree bound does not guarantee with large enough probability that some node from $I_t^u$ will send the rumor to $S$.

To tackle this problem we consider a ``restricted" rumor spreading process, on an induced subgraph of $G$.
We identify a set of nodes \emph{participating} in rumor spreading (Definition~\ref{def:participating}), such that, intuitively, each participating node has at least some constant probability to contact or be contacted by another participating node in a round.
Only nodes from  $S^+$ or $\partial(S^+)$ can be participating.
Participating nodes from $S^+$ are \emph{active}, i.e., they initiate a connection to a random neighbor in each round, while participating nodes from $\partial(S^+)$ are \emph{passive}, i.e., they \emph{accept} connections from active neighbors but do not \emph{initiate} connections to random neighbors.
Using the assumption that $S$ has low boundary expansion we show that at least some constant fraction of $\partial S$ is participating (Claim~\ref{clm:active-fraction}).
For each active node $u\in \partial S$ then, we show that an expected number of $O(\log\Delta)$ rounds suffices for a rumor originated at $u$ to spread to some node in $S$ (Lemma~\ref{lem:ppull-to-S}).
The proof of this result is now similar to that for Theorem~\ref{thm:dominating-set}, although it is crucial that we use \PPull{} rather than just \Push.

An novelty of the argument above is that it exploits the fact that in \PPull{} each edge $uv$ is chosen with probability roughly $\deg(u)^{-1} + \deg(v)^{-1}$, which is more powerful that treating push and pull operations separately.

The rest of this section is structured as follows.
In Section~\ref{sec:participating} we describe the set of participating nodes.
In Section~\ref{sec:rumor-active} we bound the time until a rumor originating from an active node reaches $S$.
And in Section~\ref{sec:finish-mainlemma} we put the pieces together to obtain Lemma~\ref{lem:growthIt}.

\subsubsection{Participating Nodes}
\label{sec:participating}

Below we give the definition of participating, active, and passive nodes, followed by some intuitive explanation.
\begin{definition}
\label{def:participating}
    The set $P$ of \emph{participating} nodes is the largest subset of $V$
    with the property that for every node $u\in P$ and for $A = P \cap S^+$,
    \begin{align}
        \label{eq:active}
        \frac{|N(u)\cap P|}{\deg(u)} + \sum_{v\in N(u)\cap A}\deg(v)^{-1}
        &\geq
        \epsilon_p,
        && \text{if $u\in A$;}
        \\
        \label{eq:passive}
        \sum_{v\in N(u)\cap A}\deg(v)^{-1}
        &\geq
        \epsilon_p,
        && \text{if $u\in P-A$,}
    \end{align}
    where $0 < \epsilon_p < (1-\epsilon_h)/3$ is a constant.\footnote{We will see later that $P$ is unique, although this is not essential for the analysis.}
    Set $A$ is the set of \emph{active} nodes, and $P - A$ is the set of \emph{passive} nodes.
\end{definition}

Note that all participating nodes belong to $S^+\cup \partial(S^+)$; the ones in $S^+$ are active, and those in $\partial(S^+)$ are passive.

Intuitively, $P$ is defined such that each participating node has at least a constant probability to contact or be contacted by another participating node in a round.
The term $|N(u)\cap P|/\deg(u)$ in Equation~\eqref{eq:active} is the probability that active node $u$ chooses a participating neighbor in a round.
In Equation~\eqref{eq:passive} we do not have this term because, as mentioned earlier, passive nodes do not initiate connections.
The sum that is common in both equations adds the probabilities of the events that $u$ is chosen by $v$, for all active neighbors $v$ of $u$.
If this sum is small (bounded by a constant), then it is of the same order as the probability that $u$ is chosen by at least one of its active neighbors, which is $1-\prod_{v\in N(u)\cap A}(1-\deg(v)^{-1})$.
(We elaborate later.)

The set $P$ can be generated by a simple procedure, which recursively removes from $V$ all nodes that do not satisfy~\eqref{eq:active} or~\eqref{eq:passive}.
Formally, we start with set $P_0 = V$.
In the $i$-th step of the procedure we obtain set $P_i$ by removing from $P_{i-1}$ all nodes $u\in A_{i-1} := P_{i-1}\cap S^+$ for which
$|N(u)\cap P_{i-1}|/\deg(u) + \sum_{v\in N(u)\cap A_{i-1}}\deg(v)^{-1}
<
\epsilon_p$, and all nodes $u\in P_{i-1} - A_{i-1}$ for which
$\sum_{v\in N(u)\cap A_{i-1}}\deg(v)^{-1} < \epsilon_p$.
Clearly, the procedure finishes after at most $n$ steps.
Let $P^\ast = P_i$ for the last step $i$.
We argue now that $P^\ast = P$ (this also proves that $P$ is unique).
From the maximality of $P$ it follows that $|P^\ast| \leq |P|$.
Thus, it suffice to show that $P \subseteq P^\ast$:
Suppose, for contradiction, that $P \nsubseteq P^\ast$, and consider the first round $i$ for which $P \nsubseteq P_{i}$.
Then, we have that $P \subseteq P_{i-1}$, and $A \subseteq A_{i-1}$, and some node $v\in P$ is removed from $P_{i-1}$ in step $i$.
If $v\in A \subseteq A_{i-1}$ then it follows from~\eqref{eq:active} that
$|N(u)\cap P_{i-1}|/\deg(u) + \sum_{v\in N(u)\cap A_{i-1}}\deg(v)^{-1}
\geq \epsilon_p$, which contradicts the assumption that $v$ is removed in step $i$.
Similarly, if $v\in P - A$, then it follows from~\eqref{eq:passive} that
$\sum_{v\in N(u)\cap A_{i-1}}\deg(v)^{-1} \geq \epsilon_p$, which again contradicts the removal of  $v$.
Thus, our assumption that  $P \nsubseteq P^\ast$ is false.

In the above procedure, if we use as a starting set $P_0$ a subset of $V$, instead of $P_0=V$, then the resulting set of participating nodes is a subset of  $P$.
We will use this observation in the proof of the claim below.

The next claim says that at least a constant fraction of the nodes
$u\in\partial S$ is participating.
\begin{claim}
    \label{clm:active-fraction}
    $|\partial S\cap P|
    \geq
    \left(1 - \frac{\epsilon_h}{(1-\epsilon_p)(1-2\epsilon_p)}\right)\cdot|\partial S|$.\footnote{From our assumption in Definition~\ref{def:participating} that $\epsilon_p < (1-\epsilon_h)/3$, it follows that $1 - \frac{\epsilon_h}{(1-\epsilon_p)(1-2\epsilon_p)}>0$.}
\end{claim}

\begin{proof}
The proof is based on a potential function argument.
We consider the procedure described above for generating $P$, but  use a smaller starting set $P_0$ as described later.
We observed earlier that such a modification yields a set of participating nodes that is a subset of $P$, thus it can only strengthen our lower bound.
The reason for this modification is that it makes the potential function we will use non-increasing.
We denote by $\Phi_i$ the potential after step $i$ of the procedure.
This potential is non-negative and is defined later.
We will show that the potential $\Phi_0$ before the first step is $\Phi_0\leq \frac{\epsilon_h}{1-\epsilon_p}\cdot |\partial S|$.
Further, we will show that for each node $u\in A_{i-1}$ that is removed from $P_{i-1}$ in step $i$, the potential decreases by at least $1-2\epsilon_p$, and the removal of a node $u\notin A_{i-1}$ decreases the potential by zero or more.
Since the potential function is non-negative, it follows that the total number of nodes $u\in S^+$ removed in all steps is at most $\Phi_0/(1-2\epsilon_p)$.
Thus, the same bound holds for the number of nodes $u\in \partial S$ removed, i.e., $|\partial S| - |\partial S\cap P|\leq \Phi_0/(1-2\epsilon_p)$.
Rearranging and using that $\Phi_0\leq \frac{\epsilon_h}{1-\epsilon_p}\cdot |\partial S|$ yields $|\partial S\cap P| \geq \big(1 - \frac{\epsilon_h}{(1-\epsilon_p)(1-2\epsilon_p)}\big)\cdot|\partial S|$.

Next we fill in the pieces omitted from the above description.
We start with the definition of the potential function.
Intuitively, the potential $\Phi_{i}$ after round $i$ measures the probability ``wasted'' in connections between participating and non-participating nodes.
Its first component, $\Phi_{i,1}$, is the sum over all $u\in A_i$ of the probability that $u$ chooses a neighbor $v\notin P_i$; the second component, $\Phi_{i,2}$, is the sum over all $u\in S^+-A_i$ of the probability that $u$ would choose a neighbor $v\in P_i$ if $u$ were active.
We give two equivalent expressions for each of $\Phi_{i,1}, \Phi_{i,2}$, to be used later on.
\begin{align*}
    \Phi_{i,1}
    &=
    \sum_{u\in A_i} \sum_{v\in N(u) - P_i} \deg(u)^{-1}
    =
    \sum_{u \notin P_i} \sum_{v\in N(u)\cap A_i} \deg(v)^{-1};
    \\
    \Phi_{i,2}
    &=
    \sum_{u\in S^+ - A_i} \sum_{v\in N(u)\cap P_i} \mspace{-18mu}\deg(u)^{-1}
    =
    \sum_{u \in P_i} \sum_{v\in N(u)\cap (S^+ - A_i)} \mspace{-42mu}\deg(v)^{-1}.
\end{align*}
Then, $\Phi_i = \Phi_{i,1} + \Phi_{i,2}$.

The new starting set $P_0$ we use consists of all $u\in S^+$, plus those $u\in \partial(S^+)$ for which
\begin{equation}
    \label{eq:P0}
    \sum_{v\in N(u)\cap \partial S}\deg(v)^{-1} \geq 2 \epsilon_p.
\end{equation}
The above condition is similar to~\eqref{eq:passive}, but the threshold is twice that in~\eqref{eq:passive}.

We can now prove that
\begin{equation}
    \label{eq:Psi0}
    \Phi_0\leq \frac{\epsilon_h}{1-\epsilon_p}\cdot |\partial S|
\end{equation}
The proof can be found in the Appendix.

It remains to show that for each node $u\in A_{i-1}$ that is removed in step $i$, the potential decreases by at least $1-2\epsilon_p$, and for each node $u\in P_{i-1} - A_{i-1}$ removed the potential does not increase.
W.l.o.g., we assume that only one node  $u\in P_{i-1}$ is removed in step $i$.
(If $k>1$ nodes should be removed we just break step $i$ into $k$ sub-steps.)

First, we consider the case in which a node $u\in P_{i-1}-A_{i-1}$ is removed in step $i$.
Then $P_i = P_{i-1} - \{u\}$ and $A_i = A_{i-1}$.
From the second expressions for $\Phi_{i,1}$ and  $\Phi_{i,2}$ we obtain
\begin{align*}
    \Phi_{i,1} - \Phi_{i-1,1}
    &=
    \sum_{v\in N(u)\cap A_i} \deg(v)^{-1}
    =:\phi_{inc},
    \\
    \Phi_{i,2} - \Phi_{i-1,2}
    &=
    - \sum_{v\in N(u)\cap (S^+ - A_{i-1})} \deg(v)^{-1}
    =: - \phi_{dcr}.
\end{align*}
Since $A_i = A_{i-1}$, we have that
$\phi_{inc} + \phi_{dcr} = \sum_{v\in N(u)\cap S^+} \deg(v)^{-1} \geq 2\epsilon_p$, from~\eqref{eq:P0}.
Further, since $u$ is removed in round $i$ it satisfies the condition $\sum_{v\in N(u)\cap A_{i-1}}\deg(v)^{-1} < \epsilon_p$, which yields $\phi_{inc} < \epsilon_p$.
From these two inequalities on $\phi_{inc}$ and $\phi_{dcr}$, it follows $\phi_{inc} < \epsilon_p < \phi_{dcr}$, and thus,
$\Phi_i - \Phi_{i-1} = \phi_{inc} - \phi_{dcr} < 0$.

Next, we consider the case in which a node $u\in A_{i-1}$ is removed.
Then $P_i = P_{i-1} - \{u\}$ and $A_i = A_{i-1}  - \{u\}$, and
\begin{align*}
    \Phi_{i,1} - \Phi_{i-1,1}
    &=
    -\underbrace{\sum_{v\in N(u) - P_{i-1}} \mspace{-5mu}\deg(u)^{-1}}_{\phi_1}
    +
    \underbrace{\sum_{v\in N(u)\cap A_i} \mspace{-5mu}\deg(v)^{-1}}_{\phi_2}
    \\&
    =
    -\phi_1 + \phi_2
    , 
    \\
    \Phi_{i,2} - \Phi_{i-1,2}
    &=
    \underbrace{\sum_{v\in N(u)\cap P_i} \mspace{-10mu}\deg(u)^{-1}}_{\phi_3}
    -
    \sum_{v\in N(u)\cap (S^+ - A_{i-1})} \mspace{-35mu}\deg(v)^{-1}
    \leq
    \phi_3
    .
\end{align*}
In the expression for $\phi_2$ we  can replace $A_i$ by $A_{i-1}$, as $u\notin N(u)$ and thus $N(u)\cap A_i = N(u)\cap A_{i-1}$.
Similarly, in the expression for $\phi_3$ we  can replace $P_i$ by $P_{i-1}$.
It follows that $\phi_1 + \phi_3 = \sum_{v\in N(u)} \deg(u)^{-1} = 1$, and also $\phi_2 + \phi_3 < \epsilon_p$ because otherwise $u$ would not be removed in round $i$.
Thus,
\begin{align*}
    \Phi_i - \Phi_{i-1}
    &=
    (\Phi_{i,1} - \Phi_{i-1,1}) + (\Phi_{i,2} - \Phi_{i-1,2})
    \\&
    \leq
    -\phi_1 + \phi_2 + \phi_3
    \leq
    -\phi_1 + 2\phi_2 + \phi_3
    \\&
    =
    -(\phi_1 +\phi_3) + 2(\phi_2 + \phi_3)
    < -1 + 2\epsilon_p
    .
\end{align*}
This completes the proof of Claim~\ref{clm:active-fraction}.
\end{proof}

\subsubsection{Spreading a Rumor from an Active Node}
\label{sec:rumor-active}

In this section we prove the following lemma, which bounds the expected time until a rumor originated at some active node $u\in \partial S$ reaches $S$.
\begin{lemma}
    \label{lem:ppull-to-S}
    Let $u\in \partial S \cap P$.
    Using \PPull, a rumor originated at $u$ spreads to at least one node $v\in S$ in an expected number of  $O(\log \Delta)$ rounds.
\end{lemma}

This result is similar to Lemma~\ref{lem:push-to-S}, but holds only for {active} nodes rather than all $u\in \partial S$, and assumes \PPull{} rather than \Push.


The proof analyzes the spread of $u$'s rumor on the subgraph induced by the set $P$ of participating nodes.
We assume that each active node (from the set $A = P\cap S^+$) chooses a random neighbor in each round, and contacts that neighbor if it is participating.
Passive nodes (from $P - A = P\cap \partial(S^+)$) do not choose neighbors; they communicate only with the active nodes that choose them in each round.

Let $I_t^u$ denote the set of informed nodes after $t$ rounds.
Similarly to the proof of Lemma~\ref{lem:push-to-S}, we define
\[
    \tau
    =
    \min\{t \colon \Exp[|I_{t+1}^u| \mid I_t^u] < (1+\varepsilon)\cdot |I_t^u| \ \vee\  |I_t^u| \geq \Delta^2 \},
\]
where $0 < \varepsilon < \epsilon_p/2$ is a constant.
The condition $|I_t^u| \geq \Delta^2$ is added because we want to show a bound of $O(\log\Delta)$, instead of $O(\log n)$ as in Lemma~\ref{lem:push-to-S}.
The square in $\Delta$ is to ensure that if at least $\Delta^2$ nodes are informed then at least $O(\Delta)$ of them are active. (We elaborate later.)
The next result is an analogue of Claim~\ref{clm:tau}.
\begin{claim}
    \label{clm:pptau}
    $\Exp[\tau] = O(\log \Delta)$.
\end{claim}

\begin{proof}
The proof is along the same lines as the proof of Claim~\ref{clm:tau}.
The main difference is that we use a new argument to lower-bound the probability that $|I_{t}^u|$ grows by a constant factor in a given round $t\leq \tau$, as the argument used in  Claim~\ref{clm:tau} for \Push{} does not extend to \PPull.

Let $X_t$, for $t\geq1$, be the 0/1 random variable that is 1 iff either
$|I_{t}^u| \geq (1+\varepsilon/3)\cdot|I_{t-1}^u|$ or $t \geq \tau$.
Further, let
\[
    \tau' = \min\bigg\{i \colon \sum_{t=1}^i X_t\geq 2\log_{1+\varepsilon/3}\Delta\bigg\}.
\]
We have
\[
    \Exp[\tau]\leq \Exp[\tau'],
\]
because: For any k, if $\tau > k$, then $|I_{k}^u| \geq (1+\varepsilon/3)^{\sum_{t=1}^k X_t}$ and $|I_t^u| < \Delta^2 $, and thus,
\[
    \sum_{t=1}^k X_t
    \leq
    \log_{1+\varepsilon/3} (|I_{k}^u|)
    < \log_{1+\varepsilon/3} (\Delta^2)
    =2\log_{1+\varepsilon/3}\Delta,
\]
which implies $\tau' > k$.
Hence, for any $k$, we have that $\tau > k$ implies $\tau' > k$, and thus $\Exp[\tau] \leq \Exp[\tau']$.

Next we establish a lower-bound on the probability that $X_t = 1$, which holds independently of the past.
We fix the outcome of the first $t-1$ rounds, and show that
\[
    \Pr(X_t = 1) \geq  1-e^{-\varepsilon^2/18}.
\]
Suppose that $\tau > t-1$.
(Otherwise, $X_t = 1$ and the inequality above  holds trivially.)
We have
\begin{align*}
    \Pr(X_t = 0)
    &=
    \Pr(|I_{t}^u| < (1+\varepsilon/3)\cdot|I_{t-1}^u|  \ \wedge\  t < \tau)
    \\&
    \leq
    \Pr(|I_{t}^u| < (1+\varepsilon/3)\cdot|I_{t-1}^u|).
\end{align*}
Also, from $\tau$'s definition and the assumption that $\tau > t-1$ it follows that
$\Exp[|I_{t}^u|] \geq (1+\varepsilon)\cdot|I_{t-1}^u|$.
If we could express $|I_{t}^u|$ as a sum of independent 0/1 random variables, then we could bound $\Pr(|I_{t}^u| < (1+\varepsilon/3)\cdot|I_{t-1}^u|)$ using the above lower bound on $\Exp[|I_{t}^u|]$ and Chernoff bounds.
We have $|I_{t}^u|=\sum_{v\in V} Y_v$, where $Y_v$ is the indicator variable of the event that $v\in I_{t}^u$.
But the random variables $Y_v$ are not independent: if two nodes $v,v'\notin I_{t-1}$ have a common neighbor $w\in I_{t-1}$, then the events that $w$ pushes the rumor to $v$ or to $v'$ in round $t$ are correlated.
However, the random variables $Y_u$ are \emph{negatively associated}~\cite{Dubhashi2009book}.
This follows from~\cite[Example 4.5]{Dubhashi2009book}.\footnote{The example cited refers to a general balls and bins model.
In our case: bins are the nodes; for each informed node $u$ we place a ball to a random neighbor of $u$; and for each uninformed node $v$ we place a ball to $v$ iff a randomly chosen neighbor of $v$ is informed. The example then shows that the number of balls in the bins are negatively associated.}
Negative association allows us to use standard Chernoff bounds,
thus we have
\begin{align*}
    \Pr(|I_{t}^u| < (1+\varepsilon/3)&\cdot|I_{t-1}^u|)
    \leq
    \Pr\!\left(|I_{t}^u| < (1+\varepsilon/3)\cdot\frac{\Exp[|I_{t}^u|]}{1+\varepsilon} \right)
    \\&
    \leq
    \Pr(|I_{t}^u| < (1-\varepsilon/3)\cdot\Exp[|I_{t}^u|])
    \\&
    \leq
    e^{-(\varepsilon/3)^2\cdot \Exp[|I_{t}^u|/2}
    \leq
    e^{-(\varepsilon/3)^2\cdot (1+\varepsilon) \cdot|I_{t-1}^u|/2}
    \\&
    \leq
    e^{-\varepsilon^2/18}.
\end{align*}
Combing the above gives $\Pr(X_t = 1) = 1 - \Pr(X_t = 0) \geq 1 - e^{-\varepsilon^2/18}$.

It follows that for any $t$, $\Pr(X_t= 1 \mid X_1\ldots X_{t-1})\geq 1-e^{-\varepsilon^2/18}$.
Thus, in the 0/1 sequence $X_1,X_2,\ldots,$ the distance between the $i$-th `1' and the $(i+1)$-th `1' is stochastically dominated by a geometric random variable with expectation $1/(1-e^{-\varepsilon^2/18})$.
From the linearity of expectation then it follows that $\Exp[\tau'] \leq (2\log_{1+\varepsilon/3}\Delta)/(1 - e^{-\varepsilon^2/18})=O(\log\Delta)$.
This completes the proof of Claim~\ref{clm:pptau}.
\end{proof}

The next result is an analogue of Claim~\ref{clm:degrees}.
It bounds the harmonic mean of the degrees of \emph{active} nodes informed in the first $\tau$ rounds.
\begin{claim}
    \label{clm:ppdegrees}
    If $\Exp[|I_{t+1}^u| \mid I_t^u] < (1+\varepsilon)\cdot |I_t^u|$
    or $|I_t^u| \geq \Delta^2$, then
    $\sum_{v\in I_t^u\cap S^+} \deg(v)^{-1} \geq  \zeta := (\epsilon_p - 2\varepsilon)/3$.
\end{claim}

\begin{proof}
First we show that $|I_t^u| \geq \Delta^2$ implies $\sum_{v\in I_t^u\cap S^+} \deg(v)^{-1} \geq \zeta$; this is the easier part.
Each informed node $w\in \partial (S^+)$ has at least one informed neighbor $v\in \partial S$, the one that pushed the rumor to $w$.
Further each node $v\in \partial S$ has at most $\deg(v) - 1\leq \Delta-1$ neighbors in $\partial (S^+)$ (and at least one in $S$).
It follows that if $|I_t^u| \geq \Delta^2$, then at least $\Delta$ of the nodes in $I_t^u$ belong to $S^+$, i.e., $|I_t^u\cap S^+|\geq\Delta$, because otherwise, we have $|I_t^u| \leq  |I_t^u \cap S^+| + |I_t^u\cap S^+|\cdot (\Delta-1) < \Delta^2$.
From this it follows that $\sum_{v\in I_t^u\cap S^+} \deg(v)^{-1}\geq |I_t^u\cap S^+|\cdot \Delta^{-1}\geq \Delta\cdot \Delta^{-1} = 1 > \zeta$.

It remains to show that $\Exp[|I_{t+1}^u| \mid I_t^u] < (1+\varepsilon)\cdot |I_t^u|$ implies $\sum_{v\in I_t^u\cap S^+} \deg(v)^{-1} \geq \zeta$.
We assume that $\sum_{v\in I_t^u\cap S^+} \deg(v)^{-1} < \zeta$, and will show that $\Exp[|I_{t+1}^u| \mid I_t^u] \geq (1+\varepsilon)\cdot |I_t^u|$.

The proof builds upon the ideas used to prove Claim~\ref{clm:degrees}.
Similarly to Claim~\ref{clm:degrees}, we count the number of uninformed nodes to which exactly one copy of the rumor is \emph{pushed} in round $t+1$ (these nodes may also receive a second copy via pull).
In addition, we also count the number of uninformed nodes that \emph{pull} the rumor in this round.
The sum of those two numbers is then a lower-bound on twice the total number of nodes informed in round $t+1$.
We lower-bound the expectation of this sum using the definition of active and passive nodes (Definition~\ref{def:participating}).

Fix $I_t^u$.
For each informed active node $v\in I_t^u\cap S^+$, let $\beta(v) = |N(v) \cap (P-I_t^u)|/ \deg(v)$ be the fraction of $v$'s neighbors that are participating and uninformed.
Then the probability that $v$ pushes the rumor to such a neighbor and no other nodes pushes the rumor to the same neighbor is lower-bounded by
\begin{align*}
    \beta(v)\cdot\!\!\! \prod_{v'\in I_t^u \cap S^+}\mspace{-15mu}
    \left(1 - \deg(v')^{-1}\right)
    &\geq
    \beta(v)\cdot \left(1 - \sum_{v'\in I_t^u \cap S^+} \mspace{-15mu} \deg(v')^{-1}\right)
    \\&
    \geq
    \beta(v)\cdot (1-\zeta).
\end{align*}
Further, for each informed node $v\in I^u_t$, let $\gamma(v) = \sum_{v'\in N(v)\cap (A - I_t^u)}\deg(v')^{-1}$ be the expected number of uninformed active nodes that pull the rumor from $v$.
It follows that the expected total number of nodes that get informed in round $t+1$ is
\begin{equation}
    \label{eq:expI-beta-gamma}
    \Exp\left[|I_{t+1}^u| - |I_{t}^u|\right]
    \geq
    \frac12\sum_{v\in I_t^u\cap S^+} \beta(v)\cdot (1-\zeta)
    +
    \frac12\sum_{v\in I_t^u} \gamma(v).
\end{equation}

Next we bound each of the above sums of $\beta(v)$ and $\gamma(v)$ using Definition~\ref{def:participating}.
For each informed active node $v\in I_t^u\cap S^+$ we apply~\eqref{eq:active}: we have
\[
    \frac{|N(v)\cap P|}{\deg(v)}
    =
    \frac{|N(v)\cap I_t^u|}{\deg(v)} +  \beta(v)
    \leq
    \frac{|I_t^u|}{\deg(v)} + \beta(v),
\]
and
\begin{align*}
    \sum_{v'\in N(v)\cap A}\deg(v')^{-1}
    &=
    \sum_{v'\in N(v)\cap (A\cap I_t^u)}\deg(v')^{-1} +  \gamma(v)
    \\&
    \leq
    \sum_{v'\in A\cap I_t^u}\deg(v')^{-1} +  \gamma(v)
    <
    \zeta +  \gamma(v).
\end{align*}
From~\eqref{eq:active} then it follows
\[
    |I_t^u|/\deg(v) + \beta(v) + \zeta +  \gamma(v) \geq \epsilon_p.
\]
Summing now over all $v\in I_t^u\cap S^+$ and rearranging yields
\begin{align*}
    \sum_{v\in I_t^u\cap S^+} (\beta(v) + \gamma(v))
    &\geq
    |I_t^u\cap S^+|\cdot (\epsilon_p-\zeta) - \sum_{v\in I_t^u\cap S^+}\frac{|I_t^u|}{\deg(v)}
    \\&
    \geq
    |I_t^u\cap S^+|\cdot (\epsilon_p-\zeta) - |I_t^u|\cdot \zeta.
\end{align*}
Next, for each informed passive node $v\in I_t^u - S^+$, we obtain similarly using~\eqref{eq:passive} that $\zeta +  \gamma(v) \geq \epsilon_p$, and summing over all such $v$ gives
\[
    \sum_{v\in I_t^u- S^+} \gamma(v)
    \geq
    |I_t^u- S^+|\cdot (\epsilon_p-\zeta).
\]
Adding the last two inequalities above yields
\begin{align*}
    \sum_{v\in I_t^u\cap S^+} \beta(v) + \sum_{v\in I_t^u}\gamma(v)
    &\geq
    |I_t^u|\cdot (\epsilon_p-\zeta) - |I_t^u|\cdot \zeta
    \\&
    =
    |I_t^u|\cdot (\epsilon_p-2\zeta).
\end{align*}
From this and~\eqref{eq:expI-beta-gamma} it follows
\[
    \Exp[|I_{t+1}^u| - |I_{t}^u|]
    \geq
    ((\zeta-1)/2)\cdot |I_t^u|\cdot (\epsilon_p-2\zeta)
    \geq
    \varepsilon\cdot |I_t^u|,
\]
where for the last inequality we used that $\zeta = (\epsilon_p - 2\varepsilon)/3$.
This completes the proof of Claim~\ref{clm:ppdegrees}.
\end{proof}

Using Claim~\ref{clm:ppdegrees} it is easy to show an $O(1)$ bound on the expected number of additional rounds after round $\tau$, until some node in $S$ gets informed.
Fix $I_\tau^u$ and suppose that $I_\tau^u\cap S=\emptyset$ (otherwise some node from $S$ is already informed).
Then $I_\tau^u\cap S^+ = I_\tau^u\cap \partial S$, and Claim~\ref{clm:ppdegrees} gives $\sum_{v\in I_t^u\cap \partial S} \deg(v)^{-1} \geq \zeta$.
Since each node $v\in I_t^u\cap \partial S$ has at least one neighbor in $S$, the probability that none of these nodes pushes the rumor to a neighbor in $S$ in a given round $t > \tau$, is at most
\[
    \prod_{v\in I_t^u\cap \partial S}\left(1 - \deg(v)^{-1}\right)
    \leq
     e^{-\sum_{v\in I_t^u\cap \partial S}\deg(v)^{-1}}
    \leq
    e^{-\zeta}.
\]
It follows that the expected number of rounds until some of the nodes $v\in I_\tau^u\cap \partial S$ pushes the rumor to a node in $S$ is upper bounded by $1/(1 - e^{-\zeta}) = O(1)$.
Combining this with Claim~\ref{clm:pptau}, which says $\Exp[\tau] = O(\log\Delta)$, proves Lemma~\ref{lem:ppull-to-S}.

\subsubsection{Finishing the Proof of Lemma~\ref{lem:growthIt}}
\label{sec:finish-mainlemma}

We will use the next standard lemma, which is the analogue of Lemma~\ref{lem:symmetry} for \PPull{}.
Versions of this result can be found, e.g.,
in~\cite{Chierichetti2010stoc,GiakkoupisS2012soda,Censor-Hillel2012stoc}.
\begin{lemma}
    \label{lem:ppsymmetry}
    Let $T(V_1,V_2)$, for $V_1,V_2\subseteq V$, be the number of rounds for \PPull{} until a rumor that is initially known to all nodes $u\in V_1$ (and only them) spreads to at least one node $v\in V_2$.
    Then, for any  $V_1,V_2\subseteq V$, random variables $T(V_1,V_2)$ and $T(V_2,V_1)$ have the same distribution.
\end{lemma}

From this lemma and Lemma~\ref{lem:ppull-to-S}, it follows that the rumor spreads from $S$ to a given node $u\in \partial S \cap P$ in an expected number of at most $\ell = O(\log\Delta)$ rounds.
Markov's Inequality then gives that in $2\ell$ rounds $u$ is informed with probability at least 1/2, and from the linearity of expectation, in $2\ell$ rounds at least 1/2 of the nodes from $\partial S \cap P$ are informed in expectation.
Using Markov's again we obtain that in $2\ell$ rounds more than $3/4$ of the nodes from $\partial S \cap P$ are still uninformed with probability at most $(1/2)/(3/4)=2/3$.
Thus, 1/4 of the nodes from $\partial S \cap P$ are informed in an expected number of most $2\ell/(1/3) = 6\ell$ steps.
From this and Claim~\ref{clm:active-fraction}, that says $|\partial S \cap P| = \Theta(|\partial S|)$, we obtain Lemma~\ref{lem:growthIt}.

\subsection{The Case of High Boundary Expansion: Analysis of the  Growth of \texorpdfstring{$\boldsymbol{I_t^+}$}{It plus}}
\label{sec:Itplus}

In this section we prove the following result.
\begin{lemma}
    \label{lem:growthItplus}
    Suppose that $I_t = S$ for some set $S \subset V$ with boundary expansion $h(S)\geq\epsilon_h$, where $ \epsilon_h > 0$ is an arbitrary constant.
    There is a positive integer $b = b(S) \leq |\partial S|/\alpha$ such that the expected number of rounds until $b$ nodes from $\partial(S^+)$ have some informed neighbor is  $O\left((b/|\partial S|)\cdot \log\Delta \right)$.
\end{lemma}

Observe that the expected number of nodes from $\partial(S^+)$ that have some informed neighbor after round $t+1$ is at least $h(S)\cdot |\partial S| = \Omega(|\partial S|)$ (from Equation~\eqref{eq:pull-bexpansion} in Section~\ref{def:bexpansion}).
To prove Lemma~\ref{lem:growthItplus} we need to bound also the variance of the number of those nodes.
Intuitively, the variance will be larger when the degrees of nodes in $\partial S$ are larger, and then larger values for $b$ are needed for the lemma to hold.
We note that the condition  $b \leq |\partial S|/\alpha$ is  to ensure that the time bound is at most $O(\log(\Delta)/\alpha)$, as this is necessary for the intended use of the lemma (in the proof of Claim~\ref{clm:psi-double}).

We give now an overview of the proof of Lemma~\ref{lem:growthItplus}.
We distinguish three cases.

The first case is when the larger contribution to the boundary expansion of $S$ is from nodes $u\in \partial S$ of degree $\deg(u) \leq c\cdot|\partial S|$, for some constant $c$.
Formally, we have  $h_{\{u\in \partial S \colon \deg (u)\leq c\cdot|\partial S| \}}(S)\geq \epsilon_h/3$ (see Definition~\ref{def:bexpansion2}).
Using the second moment method, we show that after one round of \Pull{}, we have with probability $p=\Omega(1)$ that $b = \Omega(|\partial S|)$ nodes from $\partial(S^+)$ have some informed neighbor.
It follows that $b$ nodes from $\partial(S^+)$ have some informed neighbor after an expected number of $1/p = O(1)$ rounds.


The next case is when the larger contribution to $h(S)$ comes from nodes $u\in \partial S$ of degree between $c\cdot|\partial S|$ and $|S|$.
We argue that for some $k$ from this range of degrees, the number of nodes $u\in \partial S$ with degree $k\leq \deg(u)\leq 2k$ and $\Theta(k)$ neighbors in $\partial(S^+)$ is at least $\Omega(|\partial S|/l)$, where $l = \log\frac{2\min\{|S|,\Delta\}}{\max\{c|\partial S|,\delta\}}$.
Then the probability of informing at least one such $u$ in a round of \Pull{} is $p = \Omega (|\partial S|/(kl))$.
Hence, at least $b = \Theta(k)$ nodes from $\partial(S^+)$ have some informed neighbor after an expected number of $1/p = O(bl/|\partial S|)$ rounds.

The last case is when the largest contribution to $h(S)$ is from nodes $u\in \partial S$ of degree $\deg(u) \geq d^\ast$, where $ d^\ast =\max\{|S|,c|\partial S|\}$.
We argue that $\Omega(|\partial S|)$ nodes $u\in \partial S$ have $\Omega(d^\ast)$ neighbors in $\partial(S^+)$.
As the degree of those nodes $u$ may be very large, we rely on push transmissions to inform them.
Since the degree of any node from $S$ is at most $|S^+|-1$, we argue that in one round of \Push{}, at least one of those $u$ is informed with probability $p = \Omega (|\partial S|/|S^+|)$.
Hence, at least $b = \Theta(d^\ast) = \Theta(|S^+|)$ nodes from $\partial(S^+)$ have some informed neighbor after an expected number of $1/p = O(b/|\partial S|)$ rounds.

The results we prove in the first and last cases are stronger than the statement of Lemma~\ref{lem:growthItplus}, as the time bounds shown are $O(b/|\partial S|)$ rather than $O((b/|\partial S|)\cdot \log\Delta)$.
For the second case the time bound is $O((b/|\partial S|)\cdot l)$, and $l$ can be as large as $\Theta(\log\Delta)$; there are examples for which  this bound cannot be improved to $O(b/|\partial S|)$.

The complete proof of Lemma~\ref{lem:growthItplus} is in the Appendix.

\subsection{Proof of Theorem~\ref{thm:ub-pushpull}}

\label{sec:proof-main}

We combine Lemmata~\ref{lem:growthIt} and~\ref{lem:growthItplus} to show that the expected number of rounds needed to double
$
    \Psi_t = |I_t| + |\partial I_t|/2
$
or increase $|I_t|$ above $n/2$ (whichever occurs first) is bounded by $O(\log(\Delta)/\alpha)$.
\begin{claim}
    \label{clm:psi-double}
    Let
    $
        T_t
        =
        \min\{i\colon  \Psi_{t+i} \geq 2\Psi_{t} \,\vee\, |I_{t+i}| > n/2\}.
    $
    Then
    $
        \Exp[T_t \mid I_t]
        \leq
        s\cdot \log(\Delta)/\alpha,
    $
    for some constant $s>0$.
\end{claim}

\begin{proof}
The proof is by an  inductive argument.
It relies on the following direct corollary of Lemmata~\ref{lem:growthIt} and~\ref{lem:growthItplus} (its proof is in the Appendix).
We write $I_i^S$ and $\Psi_i^S$ to denote respectively $I_{t+i}$ and $\Psi_{t+i}$
given that $I_t = S$.
\begin{corollary}
    \label{cor:growthPsi}
    For any non-empty set $S\subset V$, there is a  positive integer $b_S \leq |\partial S|/\alpha$ such that for
    \[
        \tau_S = \min\{i\colon  \Psi_{i}^S \geq \Psi_{0}^S + b_S\}
    \]
    we have
    $\Exp[\tau_S] \leq  c\cdot(b_S/|\partial S|)\cdot\log\Delta$, for some constant $c>0$.
\end{corollary}

Let
\[
    T(S,k)
    =
    \min\{i\colon  \Psi_{i}^S \geq k \,\vee\, |I_{i}^S| > n/2\}.
\]
Note that the quantity $T_t$ we are interested in can be expresses as  $T_t = T(I_t,2\Psi_t)$.
We will prove that if $|S|\leq n/2$ and $\Psi_0^S < k\leq 2\Psi_0^S$, then
\begin{equation}
    \label{eq:expTsk}
    \Exp[T(S,k)]
    \leq
    c\cdot
    \frac{4k - 3\Psi_0^S}{\alpha k}\cdot \log\Delta,
\end{equation}
where $c$ is the constant of Corollary~\ref{cor:growthPsi}.
Setting $k = 2\Psi_0^S$ yields
$\Exp[T(S,2\Psi_0^S)] \leq 5c \log(\Delta) / (2\alpha)$,
which is equivalent to  the desired inequality
$ \Exp[T_t \mid I_t = S] \leq s \log(\Delta)/\alpha$,
for $s = 5c/2$.

It remains to show~\eqref{eq:expTsk}.
The proof is by induction on $k - \Psi_0^S$.
We distinguish two cases.

\smallskip
\textbf{Case 1: $k \leq \Psi_0^S +  b_S$}.
This is the base case of the induction.
Since $k \leq \Psi_0^S +  b_S$ we have $T(S,k)\leq T(S,\Psi_0^S +  b_S) \leq \tau_S$. Also,  from Corollary~\ref{cor:growthPsi} we have $\Exp[\tau_s] \leq c\cdot(b_S/|\partial S|)\cdot\log\Delta \leq c\cdot(1/\alpha)\cdot\log\Delta$, as $b_S \leq |\partial S|/\alpha$.
Combining these two results yields
$$\Exp[T(S,k)] \leq  c\cdot(1/\alpha)\cdot \log \Delta.$$
Hence, to prove~\eqref{eq:expTsk} it suffices to show that $\frac{4k - 3\Psi_0^S}{ k}\geq1$: we have
$
    \frac{4k - 3\Psi_0^S}{ k}
    =
    4 - 3\Psi_0^S/{k}
    \geq
    1,
$
because $k > \Psi_0^S$.

\smallskip
\textbf{Case 2: $k > \Psi_0^S + b_S$}.
Let $S' = (I_{t+\tau_S}\mid (I_t = S))$.
We have $T(S,k) = T(S, \Psi_0^S + b_S) + T(S',k) = \tau_S + T(S',k)$.
Taking the expectation gives
\[
    \Exp[T(S,k)] = \Exp[\tau_S] + \Exp[T(S',k)].
\]
From Corollary~\ref{cor:growthPsi}, we have
\[
    \Exp[\tau_S] \leq c\cdot\frac{b_S}{|\partial S|}\cdot\log\Delta
\]
To bound $\Exp[T(S',k)] = \Exp[T(I_{t+\tau_S},k)\mid I_t = S]$, we observe that $\Psi_0^{S'} = (\Psi_{t+\tau_S} \mid(I_t = S)) \geq \Psi_0^{S} + b_S$.
Further, from the induction hypothesis, Inequality~\eqref{eq:expTsk} holds if we replace $S$ by $S'$ as  $k - \Psi_0^{S'} < k - \Psi_0^{S}$.
It follows
\[
    \Exp[T(S',k)]
    \leq
    c\cdot
    \frac{4k - 3(\Psi_0^S+b_S)}{\alpha k}\cdot \log\Delta.
\]
Combining the three inequalities above yields
\[
    \Exp[T(S,k)]
    \leq
    c\cdot
    \left(\frac{b_S}{|\partial S|}
    +\frac{4k - 3\Psi_0^S}{\alpha k}
    -\frac{3b_S}{\alpha k}\right)
    \cdot \log\Delta.
\]
Thus to prove~\eqref{eq:expTsk} it suffices to show that $\frac{b_S}{|\partial S|} - \frac{3b_S}{\alpha k}\leq 0$, or equivalently, that $k \leq 3 |\partial S|/\alpha$:
We use the assumption of~\eqref{eq:expTsk} that $k\leq 2\Psi_0^S$, and the fact that $|S|\leq |\partial S|/\alpha$ which follows from the definition of $\alpha$.
We have
\begin{align*}
    k
    \leq
    2\Psi_0^S
    =
    2(|S| + |\partial S|/2)
    &\leq
    2(|\partial S|/\alpha + |\partial S|/2)
    \\&
    =
    (2+\alpha)|\partial S|/{\alpha}
    \leq
    3|\partial S|/{\alpha}.
\end{align*}
This completes the proof of~\eqref{eq:expTsk}, and the proof of Claim~\ref{clm:psi-double}.
\end{proof}

Let $y = 2s \log(\Delta)/\alpha$.
From Claim~\ref{clm:psi-double} we have
$
    \Exp[T_t \mid I_t]
    \leq
    y/2,
$
and Markov's Inequality yields
$
    \Pr\left(T_t \leq y \mid I_t \right)
    \geq
    1/2.
$

We partition time into intervals of $y$ rounds, and count the number of intervals in which $\Psi_t$ doubles or $|I_t|$ exceeds $n/2$.
Formally, let $X_i$, for $i\geq 1$, be the 0/1 random variable that is 1 iff  $T_{(i-1) y} \leq y$.
From  inequality $\Pr\left(T_t \leq y \mid I_t \right) \geq 1/2$ above it follows that $\Pr(X_i = 1 \mid X_1\ldots X_{i-1}) \geq 1/2$.
Hence, $\sum _{j\leq i} X_j$ stochastically dominates the sum  of $i$ {independent} Bernoulli trials with expectation 1/2.
From Chernoff bounds then we obtain for $i^\ast = 2(\beta + 3) \log n$ that
\begin{align*}
    \Pr\bigg(\sum _{j \leq i^\ast} X_j  < \log n \bigg)
    \leq
    e^{-2(i^\ast/2 - \log n)^2/i^\ast}
    <
    n^{-\beta - 1}.
\end{align*}
This implies that $i^\ast y$ rounds suffice to inform $n/2 + 1$ nodes with probability at least $1-n^{-\beta - 1}$:
If $|I_{i^\ast y}| \leq n/2$ then fewer than $\log n$ among the $X_1\ldots X_{i^\ast}$ are 1, because $|I_{i^\ast y}| \leq n/2$ implies $\Psi_{i^\ast y} \geq 2^{\sum _{j \leq i^\ast} X_j}$, and $\Psi_t< n$ if $|I_t| < n$.
Hence, $\Pr(|I_{i^\ast y}| \leq n/2) \leq \Pr(\sum _{j \leq i^\ast} X_j  < \log n) < n^{-\beta - 1}$.

Once a set $V_1\subseteq V$ of size $|V_1|\geq n/2+1$ has been informed, any given node $u\notin V_1$ gets informed within $i^\ast y$ additional rounds with probability at least $1-n^{-\beta - 1}$:
From Lemma~\ref{lem:ppsymmetry}, the probability that $u$ learns the rumor from $V_1$ within $i^\ast y$ rounds, is the same as the probability that some node from $V_1$ learns a rumor originated at $u$ within $i^\ast y$ rounds.
And the latter probability is at least equal to the probability that $n/2+1$ nodes learn $u$'s rumor within $i^\ast y$ rounds, because then at least one of these nodes belongs to $V_1$, as $|V_1| > n/2$.

From the above and the union bound, all nodes get informed in at most $2i^\ast y = O(\log n\cdot \log(\Delta) / \alpha)$ rounds with probability at least $1-n^{-\beta}$.


\subsection*{Acknowledgements}

I would like to thank Keren Censor-Hillel and Thomas Sauerwald for helpful discussions on the motivation for this work.

\bibliography{vtight}

\clearpage
\appendix

\section{Omitted Proofs}

\subsection{Proof of Claim~\ref{clm:tau}}

Let $X_t$, for $t\geq1$, be the 0/1 random variable that is 1 iff either $|I_{t}^u| \geq (1+\epsilon/2)\cdot|I_{t-1}^u|$ or $t \geq \tau$.
Then for any $k$, we have
\[
    \Pr(\tau \leq k)
    \geq
    \Pr\left(\sum_{t=1}^k X_t \geq \log_{1+\epsilon/2} n\right),
\]
because: if $\tau > k$, then $|I_{k}^u| \geq (1+\epsilon/2)^{\sum_{t=1}^k X_t}$ and thus
$\sum_{t=1}^k X_t \leq \log_{1+\epsilon/2} (|I_{k}^u|) < \log_{1+\epsilon/2} n$; and taking the contrapositive gives that if $\sum_{t=1}^k X_t \geq \log_{1+\epsilon/2} n$ then $\tau \leq k$, which implies the desired inequality above.

Next we establish a lower-bound on the probability that $X_t = 1$, which holds independently of the past.
We fix the outcome of the first $t-1$ rounds, and show that
\[
    \Pr(X_t = 1) \geq \frac{\epsilon/2}{1-\epsilon/2}.
\]
Suppose that $\tau > t-1$.
(Otherwise, we have $X_t = 1$ and the inequality above  holds trivially.)
We bound $\Pr(X_t = 0)$ using Markov's Inequality:
\begin{align*}
    \Pr(X_t = 0)
    &
    =
    \Pr(|I_{t}^u| < (1+\epsilon/2)\cdot|I_{t-1}^u| \ \wedge\  t < \tau)
    \\&
    \leq
    \Pr(|I_{t}^u| < (1+\epsilon/2)\cdot|I_{t-1}^u|)
    \\&
    =
    \Pr(|I_{t}^u|-2\,|I_{t-1}^u| < (1+\epsilon/2)\cdot|I_{t-1}^u| - 2\,|I_{t-1}^u|)
    \\&
    =
    \Pr(2\,|I_{t-1}^u| - |I_{t}^u| > (1-\epsilon/2)\cdot|I_{t-1}^u|)
    \\&
    \leq
    \frac{2\,|I_{t-1}^u| - \Exp[|I_{t}^u|]}
    {(1-\epsilon/2)\cdot|I_{t-1}^u|}.
\end{align*}
Further, since $\tau > t-1$ we have from $\tau$'s definition that $\Exp[|I_{t}^u|] \geq (1+\epsilon)\cdot|I_{t-1}^u|$.
Hence,
\[
    \Pr(X_t = 0)
    \leq
    \frac{2\,|I_{t-1}^u| - (1+\epsilon)\cdot |I_{t-1}^u|}
        {(1-\epsilon/2)\cdot|I_{t-1}^u|}
    =
    \frac{1-\epsilon}{1-\epsilon/2},
\]
and thus $\Pr(X_t = 1) \geq \frac{\epsilon/2}{1-\epsilon/2}$.

It follows that for any $t$, $\Pr(X_t= 1 \mid X_1\ldots X_{t-1})\geq \frac{\epsilon/2}{1-\epsilon/2}$, hence, $\sum_{t=1}^k X_t$ stochastically dominates the sum of $k$ \emph{independent} 0/1 random variables with expectation $\frac{\epsilon/2}{1-\epsilon/2} = \Omega(1)$.
Thus, we can use standard Chernoff bounds to obtain $\Pr\big(\sum_{t=1}^k X_t \geq \log_{1+\epsilon/2} n\big) \geq 1-n^{-\beta}$ for $k = c\cdot \log n$, for a large enough constant $c$.
And since we saw at the beginning of the proof that $\Pr(\tau \leq k)\geq \Pr\big(\sum_{t=1}^k X_t \geq \log_{1+\epsilon/2} n\big)$, the claim follows.


\subsection{Proof of Inequality~(\ref{eq:Psi0}) from the Proof of Claim~\ref{clm:active-fraction}}

We prove that $\Phi_0\leq \frac{\epsilon_h}{1-\epsilon_p}\cdot |\partial S|$.

Let $B$ be the set of nodes from $\partial(S^+)$ that do not satisfy~\eqref{eq:P0}, i.e.,
\[
    B
    =
    \bigg\{u\in \partial(S^+)\colon \sum_{v\in N(u)\cap \partial S}\deg(v)^{-1} < 2 \epsilon_p\bigg\}.
\]
We have $P_0 = S^+\cup \left(\partial(S^+) - B\right)$ and $A_0 = S^+$.
Since $S^+- A_0=\emptyset$ it is $\Phi_{i,2} = 0$, and
\begin{align*}
    \Phi_0 = \Phi_{0,1}
    &=
    \sum_{u \notin P_0} \sum_{v\in N(u)\cap A_0} \deg(v)^{-1}
    \\&
    =
    \sum_{u\in B} \sum_{v\in N(u)\cap \partial S} \deg(v)^{-1}.
\end{align*}
From the assumption that $h(S)\leq \epsilon_h$ and Equation~\eqref{eq:hofS} (right after Definition~\ref{def:bexpansion}), we have
\begin{align*}
    \epsilon_h  |\partial S|
    \geq
    h(S)\cdot |\partial S|
    &
    =
    \sum_{u\in \partial(S^+)}\!
    \left(1-\!\prod_{v\in N(u)\cap\partial S} \left(1-\deg(v)^{-1}\right)\right)
    \\&
    \geq
    \sum_{u\in B}
    \left(1-\prod_{v\in N(u)\cap\partial S} \left(1-\deg(v)^{-1}\right)\right).
\end{align*}
For any $u\in B$,
\begin{align*}
    &\prod_{v\in N(u)\cap\partial S} \left(1-\deg(v)^{-1}\right)
    \leq
    e^{-\sum_{v\in N(u)\cap\partial S} \deg(v)^{-1}}
    \\&\qquad
    \leq
    1 -  \sum_{v\in N(u)\cap\partial S} \deg(v)^{-1}
    +  \frac12\left(\sum_{v\in N(u)\cap\partial S} \deg(v)^{-1}\right)^2
    \\&\qquad
    \leq
    1 -  \left(\sum_{v\in N(u)\cap\partial S} \deg(v)^{-1}\right)
    \cdot(1-\epsilon_p),
\end{align*}
where in the last inequality we use that for $u\in B$,  $\sum_{v\in N(u)\cap\partial S} \deg(v)^{-1}\leq 2\epsilon_p$.
Combining the above gives
\[
    \epsilon_h\cdot |\partial S|
    \geq
    (1-\epsilon_p)\cdot \sum_{u\in B}
        \left(
        \sum_{v\in N(u)\cap\partial S} \deg(v)^{-1}
        \right)
    =
    (1-\epsilon_p)\cdot\Phi_0.
\]
Thus $\Phi_0\leq \frac{\epsilon_h}{1-\epsilon_p}\cdot |\partial S|$.

\subsection{Proof of Lemma~\ref{lem:growthItplus}}

We partition $\partial S$ into three sets,
\begin{align*}
    T_1 &= \{u\in \partial S \colon \deg (u)\leq c\cdot|\partial S| \},
    \\
    T_2 &= \{u\in \partial S \colon c\cdot |\partial S| < \deg (u) \leq |S|\} \},
    \\
    T_3 &= \{u\in \partial S \colon \deg (u) > \max\{|S|,c\cdot|\partial S|\}\},
\end{align*}
where $c = (\epsilon_h/3)^2/8$.
From Definitions~\ref{def:bexpansion} and~\ref{def:bexpansion2}, it follows $h_{T_1}(S)+h_{T_2}(S)+h_{T_2}(S) \geq h(S)\geq \epsilon_h$, and thus, $h_{T_i}(S)\geq \epsilon_h/3$ for at least one $i\in\{1,2,3\}$.
The next result lower-bounds $|I_{t+1}^+ - I_{t}^+|$, i.e., the number of nodes from $\partial(S^+)$ that have an informed neighbor after one round.
\begin{claim}
    \label{clm:bexpansion}
    Let $\varepsilon = \epsilon_h/3$.
    \begin{enumerate}[(a)]
    \item
    If $h_{T_1}(S)\geq \varepsilon$ then
    $\Pr\left(|I_{t+1}^+ - I_{t}^+| \geq \varepsilon\cdot |\partial S|/2\right) \geq 1/2$.

    \item
    If $h_{T_2}(S)\geq \varepsilon$ then there is some $k\in\{c\cdot |\partial S|,\dots, |S|\}$ such that for
    $l = \log\frac{2\min\{|S|,\,\Delta\}}{\max\{c|\partial S|,\,\delta\}}\leq \log(2\Delta)$,
    \[
        \Pr\left(|I_{t+1}^+ - I_{t}^+| \geq \varepsilon  k/2 \right)
        \geq
        1 - e^{-\varepsilon\cdot |\partial S|/(4kl)}
        = \Omega\!\left(\frac{|\partial S|}{kl}\right).
    \]

    \item
    If $h_{T_3}(S)\geq \varepsilon$ then for $k = \max\{|S|,c\cdot|\partial S|\}$,
    \[
        \Pr\left(|I_{t+1}^+ - I_{t}^+| \geq \varepsilon  k/2 \right)
        \geq
        1 - e^{-\varepsilon c\cdot |\partial S|/(4k)}
        = \Omega\!\left(\frac{|\partial S|}k\right).
    \]
    \end{enumerate}
\end{claim}

\begin{proof}
\mbox{(a)$\;\,$Before} round $t+1$ starts, we fix for each node $u\in T_1$ a neighbor $v_u$ that belongs to $S$.
Let $U_{T_1}$ be the set of nodes $u\in {T_1}$ that get informed in round $t+1$ by pulling the rumor from their neighbor $u_v$.
Further, let $f_{T_1}$ be the number of nodes from $\partial(S^+)$ that have a neighbor in $U_{T_1}$.
Clearly, $|I_{t+1}^+ - I_{t}^+| \geq f_{T_1}$, thus to prove~(a) it suffices to show that
$
    \Pr\left(f_{T_1} \geq \varepsilon\cdot |\partial S|/2\right) \geq 1/2.
$
We prove this below using Chebyshev's Inequality.

Each node $u\in T_1$ belongs to $U_{T_1}$ with probability  $1/\deg(u)$ independently of the other nodes, similarly to $U_T$ in Definition~\ref{def:bexpansion2}.
It follows that $h_{T_1}(S) = \Exp[f_{T_1}]/ |\partial S|$, thus
\[
    \Exp[f_{T_1}]
    =
    h_{T_1}(S)\cdot |\partial S|
    \geq
    \varepsilon\cdot |\partial S|.
\]
For the variance of $f_{T_1}$, we will prove below that $\Var[f_{T_1}] \leq \sum_{u\in T_1}\deg(u)$.
From this it follows
\[
    \Var[f_{T_1}]
    \leq
    |T_1|\cdot (c\cdot|\partial S|)
    \leq
    c\cdot|\partial S|^2
    =
    \varepsilon^2\cdot|\partial S|^2/8,
\]
as $c = \varepsilon^2/8$.
From Chebyshev's Inequality then we obtain
\begin{align*}
    \Pr\left(f_{T_1} \leq \varepsilon\cdot |\partial S|/2\right)
    &\leq
    \Pr\big(|f_{T_1} - \Exp[f_{T_1}]| \geq \varepsilon\cdot |\partial S|/2\big)
    \\&
    \leq
    \frac{\Var[f_{T_1}]}{(\varepsilon\cdot |\partial S|/2)^2}
    \leq
    1/2,
\end{align*}
as desired.

It remains to prove  $\Var[f_{T_1}] \leq \sum_{u\in T_1}\deg(u)$.
For each $v\in \partial(S^+)$, let $X_v$ be the 0/1 random variable that is 1 iff $v$ has a neighbor in $U_{T_1}$.
Then, $f_{T_1} = \sum_{v\in \partial(S^+)} X_v$, and thus
\[
    \Var[f_{T_1}]
    =
    \sum_{(v_1,v_2)\in (\partial(S^+))^2}\Cov[X_{v_1},X_{v_2}].
\]
We have
\begin{align*}
    &\Cov[X_{v_1},X_{v_2}]
    =
    \Exp[X_{v_1}\cdot X_{v_2}] - \Exp[X_{v_1}]\cdot \Exp[X_{v_2}]
    \\&
    \qquad=
    \Pr(X_{v_1} = X_{v_2} = 1)
    - \Pr(X_{v_1} = 1)\cdot \Pr(X_{v_2} = 1).
\end{align*}
We can express $\Pr(X_{v_1}= X_{v_2}=1)$ as the sum of the following two terms:
1)~the probability that $v_1$ and $v_2$ have a common neighbor in $U_{T_1}$; by the union bound, this is at most equal to $\sum_{u\in N(v_1)\cap N(v_2)\cap T_1} (1/\deg(u))$; and 2)~the probability that each of $v_1$ and $v_2$ has a neighbor in $U_{T_1}$ but they have no common neighbors in $U_{T_1}$; this is at most equal to $\Pr(X_{v_1} = 1)\cdot \Pr(X_{v_2} = 1)$.
It follows that
$
    \Cov[X_{v_1},X_{v_2}]
    \leq
    \sum_{u\in N(v_1)\cap N(v_2)\cap T_1} (1/\deg(u)),
$
and thus
\[
    \Var[f_{T_1}]
    \leq
    \sum_{(v_1,v_2)\in (\partial(S^+))^2}
    \sum_{u\in N(v_1)\cap N(v_2)\cap T_1} (1/\deg(u)).
\]
For each node $u\in T_1$, the term $1/\deg(u)$ appears in the double sum above exactly $|N(u)\cap \partial(S^+)|^2$ times:
once for each pair $(v_1,v_2)\in \left(N(u)\cap \partial(S^+)\right)^2$.
Thus,
\begin{align*}
    \Var[f_{T_1}]
    &\leq
    \sum_{u\in T_1} (|N(u)\cap \partial(S^+)|^2/\deg(u))
    \\&
    \leq
    \sum_{u\in T_1} (\deg(u)^2/\deg(u))
    =
    \sum_{u\in T_1} \deg(u).
\end{align*}

\medskip
\noindent
\mbox{(b)$\;\,$Let} $T_2' = \{u\in T_2\colon |N(u)\cap\partial(S^+)|/\deg(u) \geq \varepsilon/2\}$ be the set of nodes $u\in T_2$ with the property that an $(\varepsilon/2)$-fraction of $u$'s neighbors belongs to $\partial(S^+)$.
We have $|T_2'|\geq \varepsilon\cdot |\partial S|/2$, for otherwise, the assumption $h_{T_2}(S) \geq \varepsilon$ is contradicted:
\begin{align*}
    h_{T_2}(S)\cdot |\partial S|
    &\leq
    \sum_{u\in T_2}\frac{|N(u)\cap\partial(S^+)|}{\deg(u)}
    \\&
    \leq
    |T_2'|\cdot 1
    +
    (|T_2| - |T_2'|)\cdot (\varepsilon/2)
    \\&
    <
    \varepsilon\cdot |\partial S|/2
    +
    |\partial S|\cdot (\varepsilon/2)
    \\&
    =
    \varepsilon\cdot |\partial S|.
\end{align*}
Since nodes in $T_2'$ have degrees in the range between $\max\{c|\partial S|,\,\delta\}$ and $\min\{|S|,\,\Delta\}$, it follows that for some $k$ in this range, at least $|T_2'|/l$ nodes $u\in T_2'$ have degree $k\leq \deg(u)\leq 2k$.
If at least one of these nodes gets informed in round $t+1$ then  $|I_{t+1}^+ - I_{t}^+| \geq \varepsilon  k/2$, and the probability that this happens is at least
\[
    1 - (1 - 1/(2k))^{|T_2'|/l}
    \geq
    1 - e^{-|T_2'|/(2lk)}
    \geq
    1 - e^{-\varepsilon\cdot |\partial S|/(4lk)}.
\]

\medskip
\noindent
\mbox{(c)$\;\,$Similarly} to~(b), we let
\[
    T_3' = \{u\in T_3\colon |N(u)\cap\partial(S^+)|/\deg(u) \geq \varepsilon/2\},
\]
and we have $|T_3'|\geq \varepsilon\cdot |\partial S|/2$.
If a node $u\in T_3'$ gets informed in round $t+1$, then we have $|I_{t+1}^+ - I_{t}^+| \geq \varepsilon k/2$, where $k = \max\{|S|,c\cdot|\partial S|\}$.
Thus, to prove the claim it suffices to show that the probability of informing at least one node $u\in T_3'$ in round $t+1$ is lower-bounded by $1 - e^{-\varepsilon c\cdot |\partial S|/(4k)}$.
Unlike the proofs for (a) and (b) which rely on pull transmissions of the rumor, we will use push transmissions here.
For each $u\in T_3'$, we fix a neighbor $v_u\in S$ of $u$, before round $t+1$ starts.
The probability that the rumor is pushed from $v_u$ to $u$ is $1/\deg(v_u)\geq 1/|S^+|$.
Further, if $i>1$ nodes $u\in T_3'$ have the same $v_u$, the probability that none of them receives the rumor via a push from $v_u$ is
$1-i/\deg(v_u) \leq (1-1/\deg(v_u))^i$, i.e., it is smaller than if the nodes had distinct neighbors $v_u$.
It follows that the probability at least one node $u\in T_3'$ receives the rumor via a push from its neighbor $v_u$ is lower-bounded by
\[
    1 - (1-1/|S^+|)^{|T_3'|}
    \geq
    1-e^{-|T_3'|/|S^+|}
    \geq
    1-e^{-(\varepsilon\cdot |\partial S|/2)/(2k/c)},
\]
where for the last inequality we used that
\[
    |S^+| = |S| + |\partial S| \leq k + k/c \leq 2k/c,
\]
as $c\leq 1$.
This completes the proof of Claim~\ref{clm:bexpansion}.
\end{proof}

From Claim~\ref{clm:bexpansion}, Lemma~\ref{lem:growthItplus} follows easily:
We will assume that before each round, all informed nodes $u\notin S$ become uninformed.
This can only decrease the number of nodes $v\in \partial(S^+)$ that have some informed neighbor after a given round $i>t$.
Further, this number becomes independent of the outcome of the previous rounds $t+1,\dots,i-1$.
From Claim~\ref{clm:bexpansion}(a) then we obtain that if $h_{T_1}(S) \geq \epsilon_h/3$, the probability that at least $\varepsilon\cdot |\partial S|/2$ nodes $v\in \partial(S^+)$ have an informed neighbor after a given round $i>t$ is at least 1/2.
It follows that the expected number of rounds until $b=\varepsilon\cdot |\partial S|/2$ nodes $v\in \partial(S^+)$ have an informed neighbor is $2 = O(b/|\partial S|)$.
Similarly, Claim~\ref{clm:bexpansion}(b) yields that if $h_{T_2}(S)\geq \epsilon_h/3$, then for some $k\in\{c\cdot |\partial S|,\dots, |S|\}$, the expected number of rounds until $b=\varepsilon k/2\leq |\partial S|/\alpha$ nodes $v\in \partial(S^+)$ have an informed neighbor is $1/\Omega(|\partial S|/kl) = O(bl/|\partial S|)$.
Finally, Claim~\ref{clm:bexpansion}(c) gives that if $h_{T_3}(S)\geq \epsilon_h/3$, then for $k = \max\{|S|,c\cdot|\partial S|\}$, the expected number of rounds until $b = \varepsilon k/2 \leq |\partial S|/\alpha$ nodes $v\in \partial(S^+)$ have an informed neighbor is $1/\Omega(|\partial S|/k) = O(b/|\partial S|)$.
This completes the proof of Lemma~\ref{lem:growthItplus}.

\subsection{Proof of Corollary~\ref{cor:growthPsi}}

If $h(S)\leq\epsilon_h < 1$ it follows from Lemma~\ref{lem:growthIt} that $b = \epsilon\cdot |\partial S|$ nodes from $\partial S$ get informed in $O((b/|\partial S|)\log\Delta)$ expected time.
If $h(S)\geq\epsilon_h > 0$ it follows from Lemma~\ref{lem:growthItplus} that $b \leq |\partial S|/\alpha$ nodes from $\partial(S^+)$ have an informed neighbor in $O((b/|\partial S|)\log\Delta)$ expected time.
In both cases $\Psi_t$ increases by at least $b_S = b/2$ in $O((b_S/|\partial S|)\log\Delta)$ expected time.

\section{Rumors Spread Fast in Graphs of Diameter 2}
\label{sec:dimeter-two}

In this section we give another example of a new result that can be proved using the machinery developed for the proof of our main result.
We show that \PPull{} completes in a logarithmic number of rounds in any graph of diameter (at most) 2.
This result can be viewed as an extension to the classic result that rumor spreading takes logarithmic time in graphs of diameter 1, i.e., in complete graphs.
Note that unlike complete graphs, some graphs of diameter 2 have very bad expansion, e.g., two cliques of the same size with one common vertex.
Further, the result does not extend to graphs of diameter 3, e.g., rumor spreading takes linear time in the dumbbell graph, which consist of two cliques of the same size and a single edge between one node from each clique.
%
\begin{theorem}
    \label{thm:degree-two}
    For any graph $G=(V,E)$ of diameter 2, \PPull{} informs all nodes of $G$ in $O(\log n)$ rounds with probability $1-O(n^{-\beta})$, for any constant $\beta>0$.
\end{theorem}
\begin{proof}[Sketch]
We fix an arbitrary pair of nodes $u,v\in V$, and show that a rumor originated at $u$ reaches $v$ in $O(\log n)$ rounds w.h.p.
We divide the rumor spreading process into three phases, each of length $c\log n$ for a sufficiently large constant $c$.

In the first phase we consider only push operations.
We define $\tau$ as in the proof of Lemma~\ref{lem:push-to-S}, and from Claims~\ref{clm:tau} and~\ref{clm:degrees} we obtain that w.h.p.\ in $O(\log n)$ rounds the rumor has spread to all nodes of some set $V_u\subseteq V$ for which $\sum_{u'\in V_u} \deg(u')^{-1} =\Omega(1)$.

In the last (third) phase we consider just pull operations.
From the symmetry between push and pull, and the argument used for the first phase, it follows that w.h.p.\ there is a set $V_v\subseteq V$ with  $\sum_{v'\in V_v} \deg(v')^{-1} =\Omega(1)$, such that if some node from $V_v$ knows the rumor then $v$ learns the rumor as well within the next $O(\log n)$ rounds.

Suppose now that sets $V_u$ and $V_v$ as above exist (this is true w.h.p.), and fix all random choices in the first and third phases (and thus sets $V_u$ and $V_v$).
We will show that in the second phase the rumor spreads from $V_u$ to some node in $V_v$ w.h.p.
This is trivially true if $V_u\cap V_v \neq \emptyset$, hence assume $V_u\cap V_v = \emptyset$.
From the assumption that the graph has diameter 2, it follows that for every node $v'\in V_{v}$, at least one of the next two conditions holds:
\begin{enumerate}
\item Node $v'$ has a neighbor in $V_u$; or

\item For every $u'\in V_{u}$ there is a node $w_{u'v'}\notin V_{u}\cup V_{v}$ that is a neighbor of both $u'$ and $v'$.
\end{enumerate}
If Condition~1 holds for all $v'\in V_{v}$, and thus each $v'\in V_{v}$ has some informed neighbor at the beginning of the second phase, then from inequality $\sum_{v'\in V_v} \deg(v')^{-1} =\Omega(1)$ it follows that  some $v'\in V_{v}$  will pull the rumor from an informed neighbor w.h.p.\ within $O(\log n)$ rounds.
Suppose now that Condition~1 does not hold for some $v'\in V_{v}$, thus Condition~2 must hold.
Then from inequality $\sum_{u'\in V_u} \deg(u')^{-1} =\Omega(1)$ it follows that w.h.p.\ at least one of the nodes $u'\in V_{u}$ will push the rumor to its neighbor $w_{u'v'}$ within $O(\log n)$ rounds.
And thus $v'$ will \emph{have} some informed neighbor after that.
Hence, after $O(\log n)$ rounds in the second phase, w.h.p.\ all $v'\in V_{v}$ have some informed neighbor.
It follows then that some $v'\in V_{v}$ will pull the rumor in an additional $O(\log n)$ rounds w.h.p., as we argued earlier.
\end{proof}

\section{Stronger Bounds for Regular Graphs}
\label{sec:regular}

In this section we focus on $\Delta$-regular graphs $G=(V,E)$.\footnote{The results we present extend also to graphs for which the ratio of maximum over minimum degree is bounded by a constant.}
For these graphs, a variant of our analysis yields an upper bound in terms of a natural new expansion measure we define.
This bound is stronger than the $O(\log(n)/\phi)$ bound with conductance.

Inspired from the proof of Theorem~\ref{thm:ub-pushpull}, and a connection between the boundary expansion $h(S)$ of a set $S$ and the conductance $\phi(\partial S)$ of its boundary (explained later), we introduce the following expansion measure.
For a non-empty set $S\subset V$ let
\[
    \xi(S) = \alpha(S)\cdot \phi(\partial S),
\]
and define
\[
    \xi(G) = \min_{S\subset V,\ 0<|S|\leq n/2} \xi(S).
\]
Then the following  bound holds.
\begin{theorem}
    \label{thm:ub-pushpull-regular}
    For any $\Delta$-regular graph $G=(V,E)$ and any constant $\beta>0$, with probability $1-O(n^{-\beta})$ \PPull{} informs all nodes of $G$ in $O(\log(n)/\xi)$ rounds.
\end{theorem}

This theorem follows from a stronger result to be described later.

It is easy to see that $\xi(S) \geq\phi(S)$:
For any set $S\subseteq V$ with $0<|S|\leq n/2$ we have
\begin{align*}
    \alpha(S)\cdot \phi(\partial S)
    &=
    \frac{|\partial S|}{|S|}\cdot \frac{E(\partial S, V - \partial S)}
    {|\partial S| \cdot \Delta}
    \\&
    =
    \frac{1}{|S|}\cdot \frac{E(\partial S, V - \partial S)}
    {\Delta}
    \\&
    \geq
    \frac{1}{|S|}\cdot \frac{E(\partial S, S)}
    {\Delta}
    =
    \phi(S).
\end{align*}
It follows that the $O(\log(n)/\xi)$ bound of Theorem~\ref{thm:ub-pushpull-regular} implies the $O(\log(n)/\phi)$ bound.

There are graphs for which the $O(\log(n)/\xi)$ bound is strictly stronger than both the $O(\log(n)/\phi)$ and the $O(\log n\cdot \log(\Delta)/\alpha)$ bounds.
Here is an example:
Graph $G$ consists of $\ell$ dense components, which are loosely connected with the other components.
We have $\ell = n/(c\Delta)$ components, and each component is a random $\Delta$-regular graph of order $c\Delta$, where $c\geq 2$ is a constant and $\Delta = \omega(\log n)$.
Inter-component edges are generated as follows,
from each node $u$ we draw an edge to a uniformly random other node.
For this graph we have $\phi = \Theta(1/\Delta)$, $\alpha = \Theta(1)$, and $\xi = \Theta(1)$.
Hence, Theorem~\ref{thm:ub-pushpull-regular} yields the correct bound of $O(\log n)$, whereas the bounds with  vertex expansion  and conductance are respectively $O(\log n \cdot \log \Delta)$ and $O(\Delta \log n)$.

We now sketch the proof of a stronger bound, which implies Theorem~\ref{thm:ub-pushpull-regular}.
For each non-empty set $S\subset V$ let
\begin{align*}
    \rho(S)
    &=
    \xi(S) + \alpha(S)/\log\Delta
    \\&
    =
    \alpha(S)\cdot(\phi(\partial S) + 1/\log\Delta),
\end{align*}
and define $\rho(G) = \min_{S\subset V,\ 0<|S|\leq n/2} \rho(S)$.
We show that an $O(\log(n)/\rho)$ bound holds w.h.p.\ for any regular graph $G$.

First we observe that
\[
    \frac{E(\partial S, \partial(S^+))}{2\Delta\cdot|\partial S|}
    \leq
    h(S)
    \leq
    \frac{E(\partial S, \partial(S^+))}{\Delta\cdot|\partial S|},
\]
where:
the right inequality holds because $E(\partial S, \partial(S^+))/\Delta$ is the expected number of edges between $\partial(S^+)$ and the set $U$ in the definition of $h(S)$ (Definition~\ref{def:bexpansion}), and this is larger than the expected number $h(S)\cdot |\partial S|$  of nodes from $\partial(S^+)$ with neighbors in $U$;
and the left inequality follows from the fact that the probability a given node $v\in \partial(S^+)$ has a neighbor in $U$ is
\[
    1-(1-1/\Delta)^{|N(v)\cap\partial S|}
    \geq
    1 - e^{-|N(v)\cap\partial S|/\Delta}
    \geq
    \frac{|N(v)\cap\partial S|}{2\Delta},
\]
where for the last relation we used the known inequality $e^{-x}\leq 1-x+x^2/2$.

The above yields the following relation between $\phi(\partial S)$ and $h(S)$,
\begin{align*}
    \phi(\partial S)
    &=
    \frac{E(\partial S, V-\partial S)}{\Delta\cdot|\partial S|}
    \\&
    =
    \frac{E(\partial S, S)}{\Delta\cdot|\partial S|}
    +
    \frac{E(\partial S, \partial(S^+))}{\Delta\cdot|\partial S|}
    \\&
    =
    \frac{E(\partial S, S)}{\Delta\cdot|\partial S|}
    +
    s\cdot h(S).
\end{align*}
where $1\leq s\leq 2$.

The proof of the $O(\log(n)/\rho)$ bound follows closely the proof of Theorem~\ref{thm:ub-pushpull}.
Given that $I_t = S$, with $|S|\leq n/2$, we now distinguishes three cases (instead of two):
\begin{enumerate}
\item
Case $\phi(\partial S) \leq 1/\log\Delta$: then $h(S) \leq 1/\log\Delta$ and we apply Lemma~\ref{lem:growthIt} to obtain that $\Omega(\partial S)$ nodes from $\partial S$ get informed in $O(\log\Delta)$ expected rounds.

\item
Case $\phi(\partial S) > 1/\log\Delta$ and $h(S) < \phi(\partial S)/4$: then $\frac{E(\partial S, S)}{\Delta\cdot|\partial S|} \geq \phi(\partial S)/2$ and a simple argument yields that $\Omega(\phi(\partial S)\cdot|\partial S|)$ nodes from $\partial S$ get informed via \Pull{} in $O(1)$ expected rounds.

\item
Case $\phi(\partial S) > 1/\log\Delta$ and $h(S) \geq \phi(\partial S)/4$:
in this case we would like to have a variant of Lemma~\ref{lem:growthItplus} giving an expected bound of $O\left((b/|\partial S|)\cdot(1/\phi(\partial S))\right)$ (instead of $O\left((b/|\partial S|)\cdot \log\Delta \right)$) on the number of rounds until $b$ nodes from $\partial(S^+)$ have some informed neighbor.
We prove that this bound holds, except for when $|\partial S|$ is close to $\Delta$ (between $(1-\epsilon)\cdot\Delta$ and $ \Delta\cdot\log(1/\phi(\partial S))\leq \Delta\cdot\log\log\Delta$).
For this special case, we show a slightly weaker time bound, which is larger by a factor of at most $\log\log\Delta$.
The proof of these results is similar to that of Lemma~\ref{lem:growthItplus}.
\end{enumerate}
If in Case~3 above we had the same bound for when $|\partial S|$ is close to $\Delta$ as for the other values of $|\partial S|$, then an analysis as in Section~\ref{sec:proof-main} would yield the desired $O(\log (n)/\rho)$ time to inform all nodes w.h.p.
We show that the weaker bounds we have when $|\partial S|$ is close to $\Delta$ affect only the bound we compute for the time to increase the number of informed nodes from $\Delta/(\alpha\log^{O(1)}\Delta)$ to $\Delta\log(\Delta)/\alpha$:
We show that $O((\log\log \Delta)^2/\rho)$ rounds are needed (instead of $O(\log\log(\Delta)/\rho)$) in expectation, and that $O(\log (n)/\rho)$ rounds suffice w.h.p.
The full proof will be given in the journal version of the paper.

\end{document}